\documentclass[a4paper,Haag duality]{mathscan}
\newtheorem{thm}{Theorem}[section] 
\newtheorem{pro}[thm]{Proposition}  
\newtheorem{cor}[thm]{Corollary}    
\newtheorem{lem}[thm]{Lemma}        
\theoremstyle{definition}           
\newtheorem{defn}[thm]{Definition}  
\newtheorem{exam}[thm]{Example}     

\newcommand{\NI}{\noindent}

\newcommand{\bea}{\begin{eqnarray}}
\newcommand{\eea}{\end{eqnarray}}

\def \b #1 {\bf #1}
\newcommand{\IR}{\mathbb{R}}

\newcommand{\IC}{\mathbb{C}}

\newcommand{\IT}{\mathbb{T}}
\newcommand{\IZ}{\mathbb{Z}}
\newcommand{\IM}{\mathbb{M}}

\newcommand{\cal}{\mathcal}
\newcommand{\clk}{{\cal K}}

\newcommand{\cla}{{\cal A}}
\newcommand{\clz}{{\cal Z}}
\newcommand{\cli}{{\cal I}}

\newcommand{\clh}{{\cal H}}

\newcommand{\clo}{{\cal O}}
\newcommand{\clb}{{\cal B}}

\newcommand{\clj}{{\cal J}}
\newcommand{\cln}{{\cal N}}

\newcommand{\clm}{{\cal M}}

\newcommand{\raro}{\rightarrow}

\newcommand{\vsp}{\vskip 1em}

\newcommand{\ul}{\underline}

\def \qed {\hfill \vrule height6pt width 6pt depth 0pt}
\newcommand{\be}{\begin{equation}}
\newcommand{\ee}{\end{equation}}
\newcommand{\ben}{\begin{eqnarray*}}
\newcommand{\een}{\end{eqnarray*}}
\begin{document}

\title{ $SU_2(\IC)$ symmetry in quantum spin chain and Haldane's conjecture}

\author{ Anilesh Mohari }
\thanks{...}

\address{
The Institute of Mathematical Sciences, \\
CIT Campus, Taramani, Chennai-600113 }

\email{anilesh@imsc.res.in}

\keywords{Uniformly hyperfinite factors. Cuntz algebra, Popescu dilation, Spontaneous symmetry breaking, Heisenberg iso-spin anti-ferromagnetic model, ground states, reflection positive, Haldane's conjecture }

\subjclass{46L}

\thanks{}

\begin{abstract}
In this paper, we prove that any translation and $SU_2(\IC)$-invariant pure state of $\IM=\otimes_{k \in \IZ}\!M^{(k)}_d(\IC)$, that is also real, lattice symmetric and reflection positive with a certain twist $r_0 \in U_d(\IC)$, is finitely correlated and its two-point spatial correlation function decays exponentially whenever $d$ is an odd integer. In particular, the Heisenberg iso-spin anti-ferromagnetic integer spin model admits unique low temperature limiting ground state and its spatial correlation function decays exponentially. The unique low temperature limiting ground state of the Hamiltonian is determined by the unique solution to Clebsch-Gordon inter-twinning isometry between two representations of $SU_2(\IC)$. 
\end{abstract}

\maketitle 

\section{ Introduction }

\vsp 
In this paper, we resume our investigation \cite{[24]} for various order properties of ground states of translation invariant Hamiltonian models \cite{[7],[31]} in the two-side infinite quantum spin chain $\IM =\otimes_{k \in \IZ}\!M_d^{(k)}(\!C)$ of the following formal form
\be 
H= \sum_{ n \in \IZ} \theta^n(h_0),
\ee
with $h^*_0=h_0 \in \IM_{loc}$, where $\IM_{loc}$ is the union of local sub-algebras of $\IM$ and $\theta$ is the right translation on $\IM$. In particular, our results are aimed to investigate the set of ground states for the Heisenberg anti-ferromagnet iso-spin model $H^{XXX}$ \cite{[6]} with nearest neighbour interactions  
\be 
h^{xxx}_0 = J (\sigma_x^0 \otimes \sigma_x^1 +\sigma_y^0 \otimes \sigma_y^1 + \sigma_z^0 \otimes \sigma_z^1),
\ee 
where $\sigma_x^k,\sigma_y^k$ and $\sigma_z^k$ are Pauli spin matrices located at lattice site $k \in \IZ$ and $J > 0$ is a constant. It is well known that any finite volume truncation of $H^{XXX}$ with periodic boundary condition admits a unique ground state \cite{[2],[6]}. However, no clear picture had emerged in the past literature about the set of ground states for the two sided infinite volume anti-ferromagnet Heisenberg $H^{XXX}$ model. Nevertheless, many interesting results on ground states, known for other specific Heisenberg type of models \cite{[19]}, such as Ghosh-Majumder (GM) model \cite{[13]} and Affleck-Kenedy-Lieb-Tasaki (AKLT) model \cite{[1]}, gave interesting conjectures on the general behaviour of ground states and its physical implication for anti-ferromagnetic Hamiltonian $H^{XXX}$ model. 

\vsp 
Haldane \cite{[2]} conjecture stated that $H^{XXX}$ has a unique ground state and the ground state admits a mass gap with its two-point spatial correlation function decaying exponentially for integer spin $s$ ( odd integer $d$, where $d=2s+1$ ). Whereas for the even values of $d$ (i.e. $s$ is a ${1 \over 2}$ odd integer spin, where $d=2s+1$), the conjecture stated that $H^{XXX}$ has a unique ground state with no mass gap and its two-point spatial correlation function does not decay exponentially. A well known result of Affleck and Lieb \cite{[2]} says: if $H^{XXX}$ admits a unique ground state for even values of $d$ then the ground state has no mass gap and its two-point spatial correlation function does not decay exponentially. In contrast, if the integer spin $H^{XXX}$ model admits a unique ground state with a mass gap, a recent result \cite{[26]} says that its two-point spatial correlation function decays exponentially. Thus the uniqueness of the hypothesis on the ground state for $H^{XXX}$ model is a critical issue to settle a part of the conjecture. 

\vsp 
For even values of $d$, we have proved in \cite{[24]} that ground state of anti-feromagnetic Heisenberg model $H^{XXX}$ is not unique and any low temperature limiting ground state of $H^{XXX}$ is not pure. Neverthelss, any low temperature limiting ground state admits no mass gap in its spectrum and its spatial correlation function does not decay exponentially provided the limiting ground state is non degenerate. In the present exposition we resume general mathematical set up \cite{[23],[24]} and address Haldane's conjecture for odd values for $d$.  

\vsp 
In the following text, we will now formulate the problem in the general framework of $C^*$-dynamical system \cite{[7],[8]} valid for two-sided one-dimensional quantum spin chain models. The uniformly hyper-finite \cite{[29]} $C^*$-algebra $\IM=\otimes_{k \in \IZ}\!M_d^{(k)}(\IC)$ of infinite tensor product of $d \times d$-square matrices $\!M_d^{(k)}(\IC) \equiv \!M_d(\IC)$, levelled by $k$ in the lattice $\IZ$ of integers, is the norm closure of the algebraic inductive limit of the net of finite dimensional $C^*$ algebras $\IM_{\Lambda}= \otimes_{k \in \Lambda }\!M_d^{(k)}(\IC)$, where $\Lambda \subset \IZ$ are finite subsets and an element $Q$ in $\IM_{\Lambda_1}$ is identified with the element $Q 
\otimes I_{\Lambda_2 \bigcap \Lambda_1^c}$ in $\IM_{\Lambda_2}$, i.e. by the inclusion map if $\Lambda_1 \subseteq \Lambda_2$, where $\Lambda^c$ is the complementary set of $\Lambda$ in $\IZ$.
We will use the symbol $\IM_{loc}$ to denote the union of all local algebras $\{ \IM_{\Lambda}: \Lambda \subset \IZ,\;|\Lambda| < \infty \}$. Thus $\IM$ is a quasi-local $C^*$-algebra with 
local algebras $\{\IM_{\Lambda}:|\Lambda| < \infty \}$ and $\IM_{\Lambda}'=\IM_{\Lambda^c}$, where $\IM'_{\Lambda}$ is the commutant of $\IM_{\Lambda}$ in $\IM$. We refer readers to Chapter 6 of [8] for more details on quasi-local $C^*$-algebras. 

\vsp 
The lattice $\IZ$ is a group under addition and for each $n \in \IZ$, we have an automorphism $\theta^n$, extending the translation action, which takes $Q^{(k)}$ to $Q^{(k+n)}$ for any $Q \in \!M_d(\IC)$ and $k \in \IZ$, by the linear and multiplicative properties on $\IM$. A unital positive linear functional $\omega$ of $\IM$ is called {\it state}. It is called {\it translation-invariant} if $\omega = \omega \theta$. A linear automorphism or anti-automorphism $\beta$ [16] on $\IM$ is called {\it symmetry } for $\omega$ if $\omega \beta = \omega$. Our primary objective is to study translation-invariant states and their symmetries that find relevance in Hamiltonian dynamics of quantum spin chain models $H$ \cite{[13],[30],[31]}.      

\vsp 
We consider \cite{[7],[30],[31]} quantum spin chain Hamiltonian in one dimensional lattice $\IM$ of the following form
\be 
H= \sum_{ n \in \IZ} \theta^n(h_0)
\ee
for $h^*_0=h_0 \in \IM_{loc}$, where the formal sum in (3) gives a group of auto-morphisms $\alpha=(\alpha_t:t \in \IR)$ by 
the thermodynamic limit: $\mbox{lim}_{\Lambda_{\eta} \uparrow \IZ}||\alpha^{\Lambda_{\eta}}_t(A)-\alpha_t(A)||=0$ for all $A \in \IM$ and $t \in \IR$ for 
a net of finite subsets $\Lambda_{\eta}$ of $\IZ$ with uniformly bounded surface energy, where automorphisms $\alpha^{\Lambda}_t(x)=e^{itH_{\Lambda}}xe^{-itH_{\Lambda}}$ is determined by the finite subset $\Lambda$ of $\IZ^k$ and $H_{\Lambda}=\sum_{n \in \Lambda} \theta^n(h_0)$. Furthermore, the limiting automorphism $(\alpha_t)$ does not depend on the net that we choose in the thermodynamic limit $\Lambda_{\eta} \uparrow \IZ$, provided the surface energies of $H_{\Lambda_\eta}$ are kept uniformly bounded. The uniquely determined group of automorphisms $(\alpha_t)$ on $\IM$ is called {\it Heisenberg flows } of $H$. In particular,  we have $\alpha_t \circ \theta^n = \theta^n \circ \alpha_t$ for all $t \in \IR$ and $n \in \IZ$. Any linear automorphism or anti-automorphism $\beta$ on $\IM_{loc}$, keeping the formal sum (3) in $H$ invariant, will also commute with $(\alpha_t)$.  

\vsp 
A state $\omega$ is called {\it stationary} for $H$ if $\omega \alpha_t= \omega$ on $\IM$ for all $t \in \IR$. The set of stationary states of $H$ is a non-empty compact convex set and has been extensively studied in the last few decades within the framework of ergodic theory for $C^*$-dynamical systems [12,Chapter 4]. However, a stationary state of $H$ need not be always translation-invariant. A stationary state $\omega$ of $\IM$ for $H$ is called $\beta$-KMS state at an inverse positive temperature $\beta > 0$ if there exists a function $z \raro f_{A,B}(z)$, analytic on the open strip $0 < Im(z) < \beta$, bounded continuous on the closed strip $0 \le Im(z) \le \beta$ with boundary condition 
$$f_{A,B}(t)=\omega_{\beta}(\alpha_t(A)B),\;\;f_{A,B}(t+i\beta)=\omega_{\beta}(\alpha_t(B)A)$$
for all $A,B \in \IM$. Using weak$^*$ compactness of convex set of states on $\clm$, finite volume Gibbs state $\omega_{\beta,\Lambda}$ is used to prove existence of a KMS state $\omega_{\beta}$ for $(\alpha_t)$ at inverse positive temperature $\beta > 0$. 
The set of KMS states of $H$ at a given inverse positive temperature $\beta$ is singleton set i.e. there is a unique $\beta$ KMS-state at a given inverse positive temperature $\beta={ 1 \over kT }$ for $H$ which has a finite range interaction [3,4,16] and thus inherits translation and other symmetry of the Hamiltonian. The unique KMS state of $H$ at a given inverse temperature is ergodic for translation dynamics. This gives a strong motivation to study translation-invariant states in a more general framework of $C^*$-dynamical systems \cite{[7]}.    

\vsp 
A state $\omega$ of $\IM$ is called {\it ground state} for $H$, if the following two conditions are satisfied:

\NI (a) $\omega(\alpha_t(A))=\omega(A)$ for all $t \in \IR$; 

\NI (b) If we write on the GNS space $(\clh_{\omega},\pi_{\omega},\zeta_{\omega})$ of $(\IM,\omega)$, $$\alpha_t(\pi_{\omega}(A))=e^{itH_{\omega}}\pi_{\omega}(A)e^{ -itH_{\omega}}$$ 
for all $A \in \IM$ with $H_{\omega}\zeta_{\omega}=0$, then $H_{\omega} \ge 0$.   

\vsp 
Furthermore, we say a ground state $\omega$ is {\it non-degenerate}, if null space of $H_{\omega}$ is spanned by $\zeta_{\omega}$ only. We say $\omega$ has a {\it mass gap}, if the spectrum $\sigma(H_{\omega})$ of $H_{\omega}$ is a subset of $\{ 0 \} \bigcap [\delta, \infty)$ for some $\delta >0$. For a wide class of spin chain models \cite{[25]}, which includes  Hamiltonian $H$ with finite range interaction, $h_0$ being in $\IM_{loc}$, the existence of a non vanishing spectral gap of a ground state $\omega$ of $H$ implies exponential decaying two-point spatial correlation functions. Now we present a precise definition for exponential decay of two-point spatial correlation functions of a state $\omega$ of $\IM$. We use symbol $\Lambda^c_m$ for complementary set of the finite volume box $\Lambda_m = \{ n: -m \le n \le m \}$ for $m \ge 1$. 

\vsp
\begin{defn} 
Let $\omega$ be a translation-invariant state of $\IM$. We say that the two-point spatial correlation functions of $\omega$ {\it decay exponentially}, if there exists a $\delta > 0$ satisfying the following condition: for any two local elements $Q_1,Q_2 \in \IM$ and $\epsilon > 0$, there exists an integer $m \ge 1$ such that    
\be
e^{\delta |n|} |\omega( Q_1 \theta^n(Q_2) ) - \omega(Q_1) \omega(Q_2)| \le \epsilon
\ee
for all $n \in \Lambda^c_m$.  
\end{defn}

\vsp 
By taking low temperature limit of $\omega_{\beta}$ as $\beta \raro \infty$, one also proves existence of a ground state for $H$ \cite{[5],[8],[29],[30]}. On the contrary to KMS states, the set of ground states is a convex face in the set of the convex set of $(\alpha_t)$ invariant states of $\IM$ and its extreme points are {\it pure} states of $\IM$ i.e. A state is called {\it pure} if it can not be expressed as convex combination of two different states of $\IM$. Thus low temperature limit points of unique $\beta-$KMS states give ground states for the Hamiltonian $H$ inheriting translation and other symmetry of the Hamiltonian. In general the set of ground states need not be a singleton set and there could be other states which are not translation invariant but still a ground state for a translation invariant Hamiltonian. Ising model admits non translation invariant ground states known as N\'{e}el state \cite{[8]}. However ground states that appear as low temperature limit of $\beta-$KMS states of a translation invariant Hamiltonian, inherit translation and other symmetry of the Hamiltonian. In particular if ground state of a translation invariant Hamiltonian model (3) is unique, then the ground state is a translation invariant pure state. 

\vsp
Let $Q \raro \tilde{Q}$ be the automorphism on $\IM$ that maps an element 
$$Q=Q_{-l}^{(-l)} \otimes Q_{-l+1}^{(-l+1)} \otimes ... \otimes Q_{-1}^{(-1)} \otimes Q_0^{(0)} \otimes Q_1^{(1)} ... \otimes Q_n^{(n)}$$ by reflecting around the point ${1 \over 2}$ of the lattice $\IZ$ to 
$$\tilde{Q}= Q_n^{(-n+1)}... \otimes Q_1^{(0)} \otimes Q_0^{(1)} \otimes Q_{-1}^{(2)} \otimes ... Q_{-l+1}^{(l)} \otimes Q_{-l}^{(l+1)}$$
for all $n,l \ge 1$ and $Q_{-l},..Q_{-1},Q_0,Q_1,..,Q_n \in 
M_d(\IC)$. 

\vsp 
For a state $\omega$ of $\IM$, we set a state $\tilde{\omega}$ of $\IM$ by 
\be 
\tilde{\omega}(Q)= \omega(\tilde{Q})
\ee
for all $Q \in \IM$. Thus $\omega \raro \tilde{\omega}$ is an affine one to one onto 
map on the convex set of states of $\IM$. The state $\tilde{\omega}$ is translation-invariant if and only if $\omega$ is translation-invariant state. We say a state $\omega$ is {\it lattice reflection-symmetric} or in short {\it lattice symmetric } if $\omega=\tilde{\omega}$.   

\vsp 
The group of unitary matrices $u \in U_d(\IC)$ acts naturally on $\IM$ as a group of automorphisms of $\IM$ defined by 
\be 
\beta_{u}(Q)=(..\otimes u \otimes u \otimes ...)Q(...\otimes u^* \otimes u^* \otimes u^*...)
\ee
We also set automorphism $\tilde{\beta}_u$ on $\IM$ defined by 
\be 
\tilde{\beta}_u(Q)=\beta_u(\tilde{Q})
\ee
for all $Q \in \IM$. So for $u,w \in U_d(\IC)$, we have 
$$\tilde{\beta}_u \tilde{\beta}_w=\beta_{uw}$$
In particular, $\tilde{\alpha}_{w}^2(Q)=Q$ for all $Q \in \IM$ if and only if $w^2=I_d$. We say a state $\omega$ of $\IM$ is {\it lattice symmetric with a twist } $w \in U_d(\IC)$ if 
\be 
w^2=I_d,\;\;\omega(\tilde{\beta}_{w}(Q))=\omega(Q)
\ee

\vsp 
We fix an orthonormal basis $e=(e_i)$ of $\!C^d$ and $Q^t \in \!M_d(\IC)$ be the transpose of $Q \in \!M_d(\IC)$ with respect to an orthonormal basis $(e_i)$ for $\IC^d$ (not complex conjugate). Let $Q \raro Q^t$ be the linear anti-automorphism map on $\IM$ that takes an element 
$$Q= Q^{(l)}_0 \otimes Q^{(l+1)}_1 \otimes ....\otimes Q^{(l+m)}_m$$
to its transpose with respect to the basis $e=(e_i)$ defined by
$$Q^t={Q^t_0}^{(l)} \otimes {Q^t_1}^{(l+1)} \otimes ..\otimes {Q^t_m}^{(l+m)},$$
where $Q_0,Q_1,...,Q_m$ are arbitrary elements in $\!M_d(\IC)$. We also note that $Q^t$ depends on the basis $e$ that we choose and avoided use of a suffix $e$. We assume that it won't confuse an attentive reader since we have fixed an orthonormal basis $(e_i)$ for our consideration through out this paper. For more general $Q \in \IM_{loc}$, we define $Q^t$ by extending linearly and take the unique bounded linear extension for any $Q \in \IM$. For a state $\omega$ of $\IM$, we define a state $\bar{\omega}$ on $\IM$ by the following prescription
\be
\bar{\omega}(Q) = \omega(Q^t)
\ee
Thus the state $\bar{\omega}$ is translation-invariant if and only if $\omega$ is translation-invariant. We say $\omega$ is {\it real }, if $\bar{\omega}=\omega$. The formal Hamiltonian $H$ is called {\it reflection symmetric with twist $w$ } if $\beta_{w}(\tilde{H})=H$ and {\it real} if $H^t=H$. 

\vsp 
We also set a conjugate linear map $Q \raro \overline{Q}$ on $\IM$ with respect to the basis $(e_i)$ for $\IC^d$ defined by extending the identity action on elements  
$$..I_d \otimes |e_{i_0}\rangle \langle e_{j_0}|^{(k)} \otimes |e_{i_1}\rangle \langle e_{j_1}|^{(k+1)} \otimes 
|e_{i_n}\rangle \langle e_{j_n}|^{(k+n)} \otimes I_d ..,\;1 \le i_k,j_k \le d,\;\;k \in \IZ,\;n \ge 0$$ 
anti-linearly. Thus by our definition we have 
$$Q^*=\overline{Q^t}$$ 
and 
$$(\overline{Q})^*=\overline{Q^*}$$

\vsp 
We set the following anti-linear reflection map $\clj_{w}:\IM \raro \IM$ with twist $w \in U_d(\IC)$, defined by 
\be 
\clj_{w}(Q) = \overline{\beta_{w}(\tilde{Q})}
\ee 
for all $Q \in \IM$. 

\vsp 
Following a well known notion [12], a state $\omega$ on $\IM$ is called {\it reflection positive with a twist $r_0 \in U_d(\IC),\;r_0^2=I_d$}, if
\be 
\omega(\clj_{r_0}(Q) Q) \ge 0
\ee 
for all $Q \in \IM_R$. Thus the notion of reflection positivity also depends explicitly on the underlining fixed orthonormal basis $e=(e_i)$ of $\IC^d$. One standard observation that we note at this point that a reflection positive with twist $r_0$ is a real state with twist $(r_0)$ after relection. Since the sesqui-linear map $(Q_1,Q_2) \raro \langle Q_1, Q_2 \rangle = \omega(\clj_{r_0}(Q_1)Q_2)$ admits polarization identity, it is skew symmetric i.e. $\langle Q_1, Q_2 \rangle = \overline{ \langle Q_2, Q_1 \rangle}$ and thus we verify the following 
identities: 
$$\omega(Q) = \overline{\omega(\clj_{r_0}(Q))}$$
$$ = \omega((\clj_{r_0}(Q))^*)$$
$$=\omega(\beta_{r_0}(\tilde{Q}^t)$$ 
$$=\bar{\omega}(\beta_{r_0}(\tilde{Q}))$$
since $Q^*=\bar{Q^t}$ i.e. $\bar{\omega}= \tilde{\omega} \beta_{r_0}$. In other words, a reflection positive translation invariant state with twist $\beta_{r_{\zeta}}$ is
always satisfy $\omega(Q)=\omega(\clj_{r_0}(Q^*))$, alternatively $\bar{\omega}=\tilde{\omega} \beta_{r_0}$. In particular, such a reflection positive state $\omega$ with twist $r_0$ is real if $\omega$ is also lattice reflection symmetric with twist $\beta_{r_0}$ invariant i.e. $\omega = \tilde{\omega} \circ \beta_{r_0}$. We will get back to this important point in section 3 in detals.

\vsp 
Let $G$ be a compact group and $g \raro u(g)$ be a $d-$dimensional unitary representation of $G$. By $\gamma_g$ we denote the product action of $G$ on the infinite tensor product $\IM$ induced by $u(g)$,
\be 
\gamma_g(Q)=(..\otimes u(g) \otimes u(g)\otimes u(g)...)Q(...\otimes u(g)^*\otimes u(g)^*\otimes u(g)^*...)
\ee
for any $Q \in \IM$, i.e. $\gamma_g=\beta_{u(g)}$. We say $\omega$ is $G$-invariant, if 
\be 
\omega(\gamma_g(Q))=\omega(Q)
\ee 
for all $Q \in \IM_{loc}$. If $G=U_d(\IC)$ and $u:U_d(\IC) \raro U_d(\IC)$ is the natural representation $u(g)=g$, then we will identify the notation $\beta_g$ with $\gamma_g$ for 
simplicity. Formal Hamiltonian $H$ given in (3) is called $G$-gauge invariant if $\gamma_g(H)=H$
for all $g \in G$.  

\vsp 
We recall now \cite{[18],[27]} if $H$ in (3) has the following form 
\be 
-H= B + \clj_{r_0}(B) + \sum_i C_i \clj_{r_0}(C_i)
\ee 
for some $B, C_i \in \IM_R$ then the unique KMS state at inverse positive temperature $\beta$ is reflection positive with the twist $r_0$. We refer to \cite{[12]} for details, which we will cite frequently while dealing with examples satisfying (12). Since the weak$^*$-limit of a sequence of reflection positive states with the twist $r_0$ is also a reflection positive state with the twist $r_0$, weak$^*$-limit points of the unique $\beta-$KMS state of $H$ as $\beta \raro \infty$, are also reflection positive with the twist $r_0$. Thus any weak$^*$ low temperature limit point ground state of $H$ is reflection positive with a twist $r_0$ if $H$ is given by (14). In particular, the unique $\beta$-KMS state of 
anti-ferromagnetic $H^{XXX}$ model is real and reflection symmetric and reflection positive with twist $r_0$ since $H^{XXX}$ admits the form (12) \cite{[12]} with $r_0=\sigma_y$. Furthermore $H^{XXX}$ admits $SU_2(\IC)$ gauge symmetry with irreducible representation $g \raro u(g)$. 
 
\vsp 
A pure mathematical question that arise here: Do these additional symmetries of $\omega$ help to understand what type of factor $\pi_{\omega}(\IM_L)''$ is? In the present exposition, as an application of our main mathematical results of \cite{[24]}, we will prove the following theorem in the forth section. 

\vsp 
\begin{thm} 
Let $\omega$ be a translation invariant, real, lattice symmetric and reflection positive with twist $r_0 \in U_d(\IC)$ state of $\IM=\otimes_{k \in \IZ}\!M^{(k)}_d(\IC)$ and the following two statements be true for odd values of $d$:

\NI (a) $\omega$ is pure;

\NI (b) $\omega$ is $SU_2(\IC)$-invariant, where $g \raro u(g) \in U_d(\IC)$ in (13) is an irreducible representation of $SU_2(\IC)$ satisfying 
\be 
r_0^2=I_d,\;\; r_0u(g)r_0^*=\bar{u(g)}
\ee
for all $g \in SU_2(\IC)$, where the matrix conjugation with respect to an orthonormal basis $e=(e_i)$ of $\IC^d$. 

\vsp 
Then $\pi_{\omega}(\IM_L)''$ is a type-I factor and two-point spatial correlation function of $\omega$ decays exponentially. 

\end{thm}  

As an application of Theorem 1.2, we will prove the following theorems in section 5.

\vsp 
\begin{thm} 
Let $H$ be a translation invariant Hamiltonian of the form $H=\sum_{ k \in \IZ} \theta_k(h_0)$ with $h_0=h_0^* \in \IM_{loc}$. 
and $H$ be also $SU_2(\IC)$ invariant with an irreducible representation $g \raro u(g)$ of $SU_2(\IC)$ and $r_0$ be the element in $U_d(\IC)$ satisfying (15). Let $H$ be also real (with respect to the basis $e=(e_i)$ ), lattice reflection symmetric and unique $\beta$-KMS at inverse positive temperature be lattice symmetric and reflection positive with the twist $r_0 \in U_d(\IC)$. 

\vsp 
Then any low temperature limiting ground state of $H$ admits ergodic decomposition, in the convex set of real, lattice reflection symmetric and translation invariant states, satisfying the following: 

\vsp 
\NI (a) If $d$ is an odd integer, then all extreme points in the ergodic decomposition are pure and $SU_2(\IC)$ invariant ground states of $H$; 

\vsp 
\NI (b) If $d$ is an even integer, then none of its extreme points in the ergodic decomposition are factor states of $\IM$ though $SU_2(\IC)$ invariant ground states of $H$.   

\end{thm} 

\vsp 
At this point we recall well known results valid for a class of $SU_2(\IC)$ invariant Hamiltonians investigated in \cite{[Babu],[Tak]} that are not of the form 
(14) and thus not evident that their unique finite temperature states admit reflection positive property. As an example, we can verify our claim easily that the Hamiltionian investigated in \cite{[Tak]} is not reflection positive with twist.       

\vsp 
Besides, in the general framework \cite{[27]}, it is also well known that mass gap 
of such a Hamiltionian in its ground state implies that its spatial correlation function decays exponentialy, though the converse statement in the general framework is not true. 
For counter examples, we refer to example 2, page 596 in \cite{[MAS]}. 

\vsp 
Neverthelss, the converse statement is likely to be true, for Hamiltontian of the form (14)  
with the additional discrete and continuous symmetry. We include a proof for the following theorem in support (but not assured ) of Haldane's conjecture in section 6.

\vsp 
\begin{thm}
Low temperature limiting ground state of anti-ferromagnet Heisenberg $H^{XXX}$ model is unique and pure for odd values of $d=3$. Moreover, the state is finitely correlated and 
its spatial two-point correlation function decays exponentially. 
\end{thm} 

\vsp  
Thus the important question that remains to be answered whether the limiting ground state $\omega_{1 \over 2}$ for integer spin $H^{XXX}$ model is having a mass gap in its spectrum from its ground state. Also note that Theorem 1.4 does not rule out possible existence of 
N\'{e}el type of ground states. 

\vsp 
The paper is organized as follows: In section 2, we will recall basic mathematical set up required from our earlier paper \cite{[24]} and explain basic ideas involved in the proof of Theorem 1.2. In section 3, we study convex set of states with various symmetries associated with positive temperature states and ground states of Hamiltionian of physical interest. Some of these results are having ready generalisation for Hamiltonians in higher lattice dimension with $SU_2(\IC)$ or more generally $SU_n(\IC)$ symmetries. In section 4, we prove Theorem 1.2. In section 5, we include a proof for Theorem 1.3 using main results of section 3 and section 4. In the last section, we will illustrate our results with models of physical interest such as $H_{GM}$, $H^{XXX}$ and $H^{AKLT}$ anti-ferromagnetic models. In particular, we will give proof of Theorem 1.4. One can use similar computation for a possible proof extending Theorem 1.4 for any odd values of $d$.

\section{Amalgamated representation of $\clo_d$ and $\tilde{\clo}_d$:}

\vsp 
A state $\psi$ on a $C^*$-algebra $\cla$ is called {\it factor}, if the center of the von-Neumann algebra $\pi_{\psi}(\cla)''$ is trivial, where $(\clh_{\psi},\pi_{\psi},\zeta_{\psi})$ is the Gelfand-Naimark-Segal (GNS) space associated with $\psi$ on 
$\cla$ \cite{[7]} and $\pi_{\psi}(\cla)''$ is the double commutant of $\pi_{\psi}(\cla)$ and $\psi(x)=\langle \zeta_{\psi}, \pi_{\psi}(x) \zeta_{\psi} \rangle$. Here we fix our convention that Hilbert spaces that are considered here are always equipped with inner products $ \langle .,. \rangle $ which are linear in the second variable and conjugate linear in the first variable. A state $\psi$ on $\cla$ is called {\it pure}, if $\pi_{\psi}(\cla)''=\clb(\clh_{\psi})$, the algebra of all bounded operators on $\clh_{\psi}$. 

\vsp
We recall that the Cuntz algebra $\clo_d ( d \in \{2,3,.., \} )$ \cite{[9]} is the universal unital $C^*$-algebra generated by the elements $\{s_1,s_2,...,
s_d \}$ subjected to the following relations:
\be 
s_i^*s_j = \delta^i_j I,\;\;\sum_{1 \le i \le d } s_is^*_i=I
\ee

\vsp
Let $\Omega=\{1,2,3,...,d\}$ be a set of $d$ elements. $\cli$ be the set of finite sequences
$I=(i_1,i_2,...,i_m)$ of elements, where $i_k \in \Omega$ and $m \ge 1$ and we use notation 
$|I|$ for the cardinality of $I$. We also include null set denoted by $\emptyset$ in the collection $\cli$ and set $s_{\emptyset }=s^*_{\emptyset}=I$ identity of $\clo_d$ and $s_{I}=s_{i_1}......s_{i_m} \in \clo_d $ and $s^*_{I}=s^*_{i_m}...s^*_{i_1} \in \clo_d$. 

\vsp
The group $U_d(\IC)$ of $d \times d$ unitary matrices acts canonically on $\clo_d$ as follows:
$$\beta_u(s_i)=\sum_{1 \le j \le d} u^j_i s_j$$
for $u=((u^i_j) \in U_d(\IC)$. In particular, the gauge action is defined by
$$\beta_z(s_i)=zs_i,\;\;z \in \IT = S^1= \{z \in \IC: |z|=1 \}.$$
The fixed point sub-algebra of $\clo_d$ under the gauge action i.e., 
$\{x \in \clo_d: \beta_z(x)=x,\;z \in S^1 \}$ is the closure of 
the linear span of all Wick ordered monomials of the form
\be 
s_{i_1}...s_{i_k}s^*_{j_k}...s^*_{j_1}:\;I=(i_1,..,i_k),J=(j_1,j_2,..,j_k)
\ee
and is isomorphic to the uniformly hyper-finite $C^*$ sub-algebra
$$\IM_R =\otimes_{1 \le k < \infty}\!M^{(k)}_d(\IC)$$
of $\IM$, where the isomorphism carries the Wick ordered monomial (20) 
into the following matrix element 
\be 
|e^{i_1}\rangle \langle e_{j_1}|^{(1)} \otimes |e^{i_2}\rangle\langle e_{j_2}|^{(2)} \otimes....\otimes |e^{i_k}\rangle\langle e_{j_k}|^{(k)} \otimes 1 \otimes 1 ....
\ee
We use notation $\mbox{UHF}_d$ for the fixed point $C^*$ sub-algebra of $\clo_d$ under the gauge group action 
$(\beta_z:z \in S^1)$. The restriction of $\beta_u$ to $\mbox{UHF}_d$ is then carried into action
$$Ad(u)\otimes Ad(u) \otimes Ad(u) \otimes ....$$
on $\IM_R$.

\vsp
We also define the canonical endomorphism $\lambda$ on $\clo_d$ by
\be 
\lambda(x)=\sum_{1 \le i \le d}s_ixs^*_i
\ee
and the isomorphism carries $\lambda$ restricted to $\mbox{UHF}_d$ into the one-sided shift
$$y_1 \otimes y_2 \otimes ... \raro 1 \otimes y_1 \otimes y_2 ....$$
on $\IM_R$. We note for all $u \in U_d(\IC)$ that $\lambda \beta_u = \beta_u \lambda$ 
on $\clo_d$ and so in particular, also on $\mbox{UHF}_d$. 

\vsp
Let $\omega'$ be a $\lambda$-invariant state on the $\mbox{UHF}_d$ sub-algebra of $\clo_d$. Following \cite{[14]}, section 7] and $\omega$ be the inductive limit state $\omega$ of $\IM \equiv \tilde{\mbox{UHF}}_d \otimes \mbox{UHF}_d$. In other word $\omega'=\omega_R$ once we make the identification $\mbox{UHF}_d$ with $\IM_R$.
We consider the set 
$$K_{\omega}= \{ \psi: \psi \mbox{ is a state on } \clo_d \mbox{ such that } \psi \lambda =
\psi \mbox{ and } \psi_{|\mbox{UHF}_d} = \omega_R \}$$
By taking invariant mean on an extension of $\omega_R$ to $\clo_d$, we verify that $K_{\omega}$ is non empty and 
$K_{\omega}$ is clearly convex and compact in the weak topology. In case $\omega$ is an ergodic state ( extremal state ) then, $\omega_R$ is as well an extremal state in the set of $\lambda$-invariant states of $\IM$. Thus
$K_{\omega}$ is a face in the $\lambda$ invariant states. Now we recall Lemma 7.4 
of \cite{[8]} in the following proposition which quantifies what we can gain 
by considering a factor state on $\clo_d$ instead of its restriction to $\mbox{UHF}_d$.

\vsp 
Our next two propositions are adapted from results in section 6 and section 7 of \cite{[8]} as stated in the present form in Proposition 2.5 and Proposition 2.6 in \cite{[23]}. 

\vsp 
\begin{pro} 
Let $\omega$ be a translation invariant ergodic state of $\IM$ then $K_{\omega}$ is a face in the convex set of $\lambda$-invariant states of $\clo_d$. Moreover the following holds:

\vsp 
\NI (a) An element $\psi \in K_{\omega}$ is ergodic if and only if $\psi$ is a factor state. Furthermore, any other extremal point in $K_{\omega}$ is of the form $\psi \beta_z$ for some $z \in S^1$;

\vsp 
\NI (b) The close subgroup $H=\{z \in S^1: \psi \beta_z =\psi \}$ is independent of the extremal point 
$\psi \in K_{\omega}$ of our choice;   

\vsp 
\begin{proof} 
For the proof for (a) and (b), we refer to Lemma 7.4 in \cite{[8]}. 
\end{proof}

\end{pro} 

\vsp 
\begin{pro} 
Let $\psi$ be a $\lambda$ invariant ergodic state on $\clo_d$ and $(\clh_{\psi},\pi_{\psi},\zeta_{\psi})$ be its GNS representation. Then the following holds:

\vsp
\NI (a) The closed subgroup $H=\{z \in S^1: \psi \beta_z =\psi \}$ is equal to 

$$\{z \in S^1: \beta_z \mbox{extends to an automorphism of } \pi_{\psi}(\clo_d)'' \} $$ 

\vsp 
\NI (b) Let $\clo_d^{H}$ be the fixed point sub-algebra in $\clo_d$ under the gauge group $\{ \beta_z: z \in H \}$. Then  
$\pi_{\psi}(\clo_d^{H})'' = \pi_{\psi}(\mbox{UHF}_d)''$. 

\vsp 
\NI (c) Let $\omega'$ be a $\lambda$-invariant state of $\mbox{UHF}_d$ algebra and  
$\pi_{\omega'}(\mbox{UHF}_d)''$ is a type-I factor, then there exists a $\lambda$-invariant factor state $\psi$ on $\clo_d$ extending $\omega'$ such that 
$$\pi_{\psi}(\mbox{UHF}_d)'' = \pi_{\psi}(\clo_d)''$$   
\end{pro} 

\vsp 
\begin{proof}
For a proof, we refer to Proposition 2.2 in \cite{[24]}. 
\end{proof} 

\vsp 
\begin{pro} 
Let $(\clh_{\psi},\pi_{\psi},\zeta_{\psi})$ be the GNS representation of a $\lambda$ invariant state $\psi$ on $\clo_d$ and $P$ be the support projection of the normal state $\psi_{\zeta_{\psi}}(X)=\langle\zeta_{\psi},X\zeta_{\psi}\rangle$ in the 
von-Neumann algebra $\pi_{\psi}(\clo_d)''$. Then the following holds:

\vsp
\NI (a) $P$ is a sub-harmonic projection for the endomorphism $\Lambda(X)=\sum_k S_kXS^*_k$ on $\pi_{\psi}(\clo_d)''$
i.e. $\Lambda(P) \ge P$ satisfying the following:

\NI (i) $PS^*_kP=S^*_kP,\;\;1 \le k \le d$;

\NI (ii) The set $\{ S_If: Pf=f,\;f \in \clh_{\psi}, |I| < \infty \}$ is total in $\clh_{\psi}$; 

\NI (iii) $\Lambda_n(P) \uparrow I$ as $n \uparrow \infty$;

\NI (iv) $\sum_{1 \le k \le d} v_kv_k^*=I_{\clk};$ 

\NI where $S_k=\pi_{\psi}(s_k)$ and $v_k=PS_kP$ for $1 \le k \le d$ are contractive operators on Hilbert subspace $\clk$, the range of the projection $P$;

\vsp
\NI (b) For any $I=(i_1,i_2,...,i_k),J=(j_1,j_2,...,j_l)$ with $|I|,|J| < \infty$ we have $\psi(s_Is^*_J) =
\langle \zeta_{\psi},v_Iv^*_J\zeta_{\psi}\rangle$ and the vectors $\{ S_If: f \in \clk,\;|I| < \infty \}$ are total in $\clh_{\psi}$;

\vsp
\NI (c) The von-Neumann algebra $\clm=P\pi_{\psi}(\clo_d)''P$, acting on the Hilbert space
$\clk$ i.e. range of $P$, is generated by $\{v_k,v^*_k:1 \le k \le d \}''$ and the normal state
$\phi(x)=\langle\zeta_{\psi},x \zeta_{\psi}\rangle$ is faithful on the von-Neumann algebra $\clm$.

\vsp
\NI (d) The following statements are equivalent:

\NI (i) $\psi$ is a factor state of $\clo_d$;

\NI (ii) $\clm$ is a factor;
  
\end{pro} 

\vsp 
\begin{proof}
For a proof we refer to Proposition 2.1 in \cite{[24]}.
\end{proof}

\vsp 
Let $\psi$ be a $\lambda$-invariant state of $\clo_d$ as in Proposition 2.2 and 
$H=\{z \in S^1: \psi = \psi \beta_z \}$ be the closed subgroup of $S^1$. Let $z \raro U_z$ be the unitary representation of $H$ in the GNS space $(\clh_{\psi},\pi,\zeta_{\psi})$ associated with the state $\psi$ of 
$\clo_d$, defined by 
\be 
U_z\pi_{\psi}(x)\zeta_{\psi}=\pi_{\psi}(\beta_z(x))\zeta_{\psi}
\ee
so that $\pi_{\psi}(\beta_z(x))=U_z\pi_{\psi}(x)U_z^*$ for $x \in \clo_d$. We use same notations $(\beta_z:z \in H)$ for its normal extensions as group of automorphisms on $\pi_{\psi}(\clo_d)''$. Furthermore, $\langle \zeta_{\psi},P\beta_z(I-P)P \zeta_{\psi} \rangle=0$ as $\psi=\psi \beta_z$ for $z \in H$. Since $P$ is the support projection of $\psi$ in $\pi_{\psi}(\clo_d)''$, we have $P\beta_z(I-P)P=0$ i.e. $\beta_z(P) \ge P$ for all $z \in H$. Since $H$ is a group, we conclude that $\beta_z(P)=P$. So $P \in \pi_{\psi}(\mbox{UHF}_d)''$ by Proposition 2.2 (b).

\vsp
Since $\phi$ is a faithful state of $\clm$, $\zeta_{\phi}$ once identified with 
$\zeta_{\psi} \in \clk$ is a cyclic and separating vector for $\clm$ and 
the closure of the closable operator $S_0:a\zeta_{\phi} \raro a^*\zeta_{\phi},\;a \in \clm, S$ possesses a polar decomposition $S=\clj \Delta^{1/2}$, where $\clj$ is an anti-unitary and $\Delta$ is a non-negative self-adjoint operator on $\clk$. M. Tomita \cite{[7]} theorem says that $\Delta^{it} \clm \Delta^{-it}=\clm,\;t \in \IR$ and $\clj \clm \clj=\clm'$, where $\clm'$ is the commutant of $\clm$. We define the modular automorphism group
$\sigma=(\sigma_t,\;t \in \IT )$ on $\clm$
by
$$\sigma_t(a)=\Delta^{it}a\Delta^{-it}$$ which satisfies the modular relation
$$\phi(a\sigma_{-{i \over 2}}(b))=\phi(\sigma_{{i \over 2}}(b)a)$$
for any two analytic elements $a,b$ for the group of automorphisms $(\sigma_t)$. A more useful modular relation used frequently in this paper is given by 
\be 
\phi(\sigma_{-{i \over 2}}(a^*)^* \sigma_{-{i \over 2}}(b^*))=\phi(b^*a)
\ee 
which shows that $\clj a\zeta_{\phi}= \sigma_{-{i \over 2}}(a^*)\zeta_{\phi}$ for an analytic element $a$ for the automorphism group $(\sigma_t)$. Anti unitary operator $\clj$ and the group of automorphism $\sigma=(\sigma_t,\;t \in \IR)$ are called {\it conjugate operator} and {\it modular automorphisms } associated with $\phi$ respectively. 

\vsp 
The state $\phi(a)= \langle \zeta_{\phi},x \zeta_{\phi} \rangle $ on $\clm$ being faithful and invariant of $\tau:\clm \raro \clm$, we find a unique unital completely positive map 
$\tilde{\tau}:\clm' \raro \clm'$ ([section 8 in \cite{[27]} ) satisfying the duality relation 
\be 
\langle b\zeta_{\phi},\tau(a)\zeta_{\phi} \rangle =  \langle \tilde{\tau}(b)\zeta_{\phi},a\zeta_{\phi} \rangle 
\ee
for all $a \in \clm$ and $b \in \clm'$. For a proof, we refer 
to section 8 in the monograph \cite{[27]} or section 2 in \cite{[22]}. 

\vsp 
Since $\tau(a)=\sum_{1 \le k \le d} v_kav_k^*,\;x \in \clm$ is an {\it inner map } i.e. each $v_k \in \clm$, we have an explicit formula for $\tilde{\tau}$ as follows: For 
each $1 \le k \le d$, we set contractive operator 
\be 
\tilde{v}_k = \overline{ \clj \sigma_{i \over 2}(v^*_k) \clj } \in \clm'
\ee 
That $\tilde{v}_k$ is indeed well defined as an element in $\clm'$ given in section 8 in \cite{[8]}. By the modular relation (21), we have  
\be  
\sum_k \tilde{v}_k \tilde{v}_k^*=I_{\clk}\;\;\mbox{and}\;\;
\tilde{\tau}(b)=\sum_k \tilde{v}_kb\tilde{v}^*_k,\; b \in \clm' 
\ee
Moreover, if $\tilde{I}=(i_n,..,i_2,i_1)$ for $I=(i_1,i_2,...,i_n)$, we have 
$$\tilde{v}^*_I\zeta_{\phi}$$
$$=\clj \sigma_{i \over 2}(v_{\tilde{I}})^*\clj\zeta_{\phi}$$
$$= \clj \Delta^{1 \over 2}v_{\tilde{I}}\zeta_{\phi}$$
$$=v^*_{\tilde{I}}\zeta_{\phi}$$
and    
\be
\phi(v_Iv^*_J)= \phi(\tilde{v}_{\tilde{I}}\tilde{v}^*_{\tilde{J}}),\; 
|I|,|J| < \infty 
\ee
We also set $\tilde{\clm}$ to be the von-Neumann algebra generated by $\{\tilde{v}_k: 1 \le k \le d \}$. Thus $\tilde{\clm} \subseteq \clm'$.

\vsp
Following \cite{[8]} and \cite{[22]}, we consider the amalgamated tensor product $\clh \otimes_{\clk} \tilde{\clh}$ of $\clh$ with 
$\tilde{\clh}$ over the joint subspace $\clk$. It is the completion of the quotient of the set 
$$\IC \bar{I} \otimes \IC I \otimes \clk,$$ 
where $\bar{I},I$ both consisting of all finite sequences with elements in $\{1,2, ..,d \}$, by the equivalence relation 
defined by a semi-inner product defined on the set by requiring
$$ \langle \bar{I} \otimes I \otimes f,\bar{I}\bar{J} \otimes IJ \otimes g \rangle = \langle f,\tilde{v}_{\bar{J}}v_Jg \rangle, $$
$$ \langle \bar{I}\bar{J} \otimes I \otimes f, \bar{I} \otimes IJ \otimes g \rangle  = \langle \tilde{v}_{\bar{J}}f,v_Jg \rangle $$
and all inner product that are not of these form are zero. We also define two 
commuting representations $(S_i)$ and $(\tilde{S}_i)$ of $\clo_d$ on
$\clh \otimes_{\clk} \tilde{\clh}$ by the following prescription:
$$S_I\lambda(\bar{J} \otimes J \otimes f)=\lambda(\bar{J} \otimes IJ \otimes f),$$
$$\tilde{S}_{\bar{I}}\lambda(\bar{J} \otimes J \otimes f)=\lambda(\bar{J}\bar{I} \otimes J \otimes f),$$
where $\lambda$ is the quotient map from the index set to the Hilbert space. Note that the subspace generated by
$\lambda(\emptyset \otimes I \otimes \clk)$ can be identified with $\clh$ and earlier $S_I$ can be identified
with the restriction of $S_I$, defined here. Same is valid for $\tilde{S}_{\bar{I}}$. The subspace $\clk$ is
identified here with $\lambda(\emptyset \otimes \emptyset \otimes \clk)$. 
Thus $\clk$ is a cyclic subspace for the representation $$\tilde{s}_j \otimes s_i \raro \tilde{S}_j S_i$$ 
of $\tilde{\clo}_d \otimes \clo_d$ in the amalgamated Hilbert space. Let $P$ be the projection on $\clk$. Then we have 
$$S_i^*P=PS_i^*P=v_i^*$$
$$\tilde{S}_i^*P=P\tilde{S}_i^*P=\tilde{v}^*_i$$
for all $1 \le i \le d$. 

\vsp 
We sum up required results in the following proposition.

\vsp 
\begin{pro} Let $\psi$ be an element in $K_{\omega}$ and $(\clk,v_k,\;1 \le k \le d)$ be the elements in the support projection of $\psi$ in $\pi_{\psi}(\clo_d)''$ described in Proposition 2.3 and $(\clk,\tilde{v}_k,\;1 \le k \le d)$ be the dual elements and $\pi$ be the amalgamated representation of $\tilde{\clo}_d \otimes \clo_d$. Then the following holds:

\vsp 
\NI (a) For any $1 \le i,j \le d$ and $|I|,|J|< \infty$ and $|\bar{I}|,|\bar{J}| < \infty$
$$ \langle \zeta_{\psi},\tilde{S}_{\bar{I}}\tilde{S}^*_{\bar{J}} S_iS_IS^*_JS^*_j \zeta_{\psi} \rangle = \langle \zeta_{\psi}, 
\tilde{S}_i \tilde{S}_{\bar{I}}\tilde{S}^*_{\bar{J}}\tilde{S}^*_jS_IS^*_J \zeta_{\psi} \rangle ;$$

\vsp 
\NI (b) The state $\psi: x \raro \langle \zeta_{\psi},x \zeta_{\psi} \rangle$ defined on  $\tilde{\mbox{UHF}}_d \otimes \mbox{UHF}_d$ is equal to $\omega$ on $\IM$, where we have 
identified   
$$\IM \equiv \IM_{(-\infty, 0]} \otimes \IM_{[1,\infty)} \equiv 
\tilde{\mbox{UHF}}_d \otimes \mbox{UHF}_d;$$ 
with respect to an orthonormal basis $e=(e_i)$ of $\IC^d$. 

\vsp 
If $\omega$ is an ergodic state of $\IM$ and $\psi$ is an extremal element in $K_{\omega}$ then 

\vsp 
\NI (c) $\pi(\clo^{H}_d))''=\pi(\mbox{UHF}_d)''$ and $\pi(\tilde{\clo}^{H}_d)''=\pi(\mbox{UHF}_d)''$;

\vsp 
\NI (d) The following statements are equivalent:

\vsp 
\NI (i) $\omega$ is a factor state of $\clo_d$;

\vsp 
\NI (ii) $\pi(\tilde{\clo}_d \otimes \clo_d)''= \clb(\tilde{\clh} \otimes_{\clk} \clh)$ 

\vsp 
\NI (iii) $\clm \vee \tilde{\clm} = \clb(\clk)$.

\end{pro}

\vsp 
\begin{proof}
For a proof we refer to Proposition 3.1. in \cite{[24]}. 
\end{proof}

\vsp
Let $G$ be a compact group and $g \raro u(g)$ be a $d-$dimensional unitary representation of $G$. By $\gamma_g$, we denote
the product action of $G$ on the infinite tensor product $\IM$ induced by $u(g)$,
$$\gamma_g(Q)=(..\otimes u(g) \otimes u(g)\otimes u(g)...)Q(...\otimes u(g)^*\otimes u(g)^*\otimes u(g)^*...)$$
for any $Q \in \IM$. We recall now that the canonical action of the group $U_d(\IC)$ of $d \times d$ matrices on
$\clo_d$ is given by
$$\beta_{u(g)}(s_j)=\sum_{1 \le i \le d} s_i u(g)^i_j $$
and thus
$$\beta_{u(g)}(s^*_j) = \sum_{1 \le i \le d} \bar{u(g)^i_j} s^*_i $$

\vsp
Note that $u(g)|e_i><e_j|u(g)^*=|u(g)e_i><u(g)e_j| = \sum_{k,l} u(g)^l_i \bar{u(g)}^k_j|e_l><e_k|$, where $e_1,..,e_d$
are the standard basis for $\IC^d$. Identifying $|e_i><e_j|$ with $s_is^*_j$, we verify that on $\IM_R$ the gauge action
$\beta_{u(g)}$ of the Cuntz algebra $\clo_d$ and $\gamma_g$ coincide i.e. $\gamma_g(Q)=\beta_{u(g)}(Q)$
for all $Q \in \IM_R$.

\vsp
\begin{pro} 
Let $\omega$ be a translation invariant ergodic state on $\IM$. Suppose that
$\omega$ is $G-$invariant,
$$\omega(\gamma_g(Q))=\omega(Q) \; \mbox{for all } g \in G \mbox{ and any } Q \in \IM. $$
Let $\psi$ be an extremal point in $K_{\omega}$ and $(\clk,\clm,v_k,\;1 \le k \le d,\phi)$ be the elements associated with $(\clh, S_i=\pi(s_i),\zeta_{\psi})$, described as in Proposition 2.3. Then we have the following:

\vsp 
\NI (a) There exists a unitary representation $g \raro \hat{U}(g)$ in $\clb(\tilde{\clh} \otimes_{\clk} \clh)$ and a representation $g \raro \zeta(g) \in S^1$ 
so that
\be
\hat{U}(g)S_i\hat{U}(g)^*= \zeta(g) \beta_{u(g)}(S_i),\;1 \le i \le d
\ee
and
\be 
\hat{U}(g)\tilde{S}_i\hat{U}(g)^*= \zeta(g) \beta_{u(g)}(\tilde{S}_i),\;1 \le i \le d
\ee
for all $g \in G$. 

\vsp 
\NI (b) There exists a unitary representation $g \raro \hat{u}(g)$ in $\clb(\clk)$ so that $\hat{u}(g)\clm \hat{u}(g)^*=\clm$ for
all $g \in G$ and $\phi(\hat{u}(g)x\hat{u}(g)^*)=\phi(x)$ for all $x \in \clm$. Furthermore, the operator
$V^*=(v^*_1,..,v^*_d)^{tr}: \clk \raro \IC^d \otimes \clk $ is an isometry which intertwines the representation
of $G$,
\be
( \zeta(g) \hat{u}(g) \otimes u(g) )V^*=V^*\hat{u}(g)
\ee
for all $g \in G$, where $g \raro \zeta(g)$ is the representation of $G$ in $U(1)$.

\vsp 
\NI (c) $\clj \hat{u}(g) \clj=\hat{u}(g)$ and $\Delta^{it}\hat{u}(g)\Delta^{-it}=\hat{u}(g)$ for all $g \in G$ and $t \in \IR$.

\vsp 
\NI (d) $u_z\hat{u}(g) = \hat{u}(g)u_z$ for all $g \in G$ and $z \in H$. 

\vsp 
\NI (e) If $G$ is simply connected then $\zeta(g)=1$ for all $g \in G$ and $\psi = \psi \beta_{u(g)}$ and $\psi_0=\psi_0 \beta_{u(g)}$ for all $g \in G$. 

\end{pro}

\vsp 
\begin{proof}
Proof is given in Proposition 2.7 in \cite{[24]}, where we used factor property of $\psi$ but same holds good if $\omega$ is an ergodic state of $\IM$ once we use Proposition 2.1 (a) instead of factor property of $\omega$.
\end{proof}

\vsp 
For a given $u \in U_d(\IC)$, we also extend the map $\clj_u:\IM \raro \IM$ defined in (10) to an anti-linear automorphism on 
$\tilde{\clo}_d \otimes \clo_d$, defined by 
\be 
\clj_u(\tilde{s}_{I'}\tilde{s}^*_{J'} \otimes s_Is_J^*) = \beta_{\bar{u}}(\tilde{s}_I\tilde{s}_J \otimes s_{I'}s^*_{J'})
\ee 
for all $|I|,|J|,|I'|,|J'| < \infty$ and then extend anti-linearly for an arbitrary element of $\tilde{\clo}_d \otimes \clo_d$. So we have 
\be 
\clj_u = \beta_{\bar{u}} \clj_{_{I_d}}= \clj_{_{I_d}} \beta_{u}
\ee 

\vsp 
So these maps are defined after fixing the orthonormal basis $e=(e_i)$, which have identified $\tilde{\mbox{UHF}} \otimes \mbox{UHF}_d$ with $\IM_L \otimes \IM_R= \IM$ as in Proposition 2.4 (b), where the monomial given (17) is identified with the matrix given in (18).  

\vsp  
We make few simple observations in the following for $u,w \in U_d(\IC)$:
$$\clj_{u}\clj_{w}$$
$$=\beta_{\bar{u}}\clj_{_{I_d}} \clj_{_{I_d}} \beta_w$$ 
\be 
= \beta_{\bar{u}w},
\ee  
and 

$$\clj_{w}\beta_{u}$$
$$=\clj_{_{I_d}} \beta_w \beta_u$$
$$=\clj_{_{I_d}} \beta_{wu}$$
\be 
=\clj_{wu}
\ee 

Also
$$\beta_u \clj_{w}$$
$$=\beta_u \clj_{_{I_d}} \beta_w$$
$$=\clj_{_{I_d}}\beta_{\bar{u}}\beta_w$$
$$=\clj_{_{I_d}}\beta_{\bar{u}w}$$
\be 
=\clj_{\bar{u}w}
\ee
i.e. $\clj_w$ commutes with $\beta_u$ if $wuw^*=\bar{u}$. 

\vsp
By combining relations (32)-(33), we have the following identities
$$\clj_w\beta_{u}$$
$$=\clj_{wu}\; \mbox{by}\; (32)$$
$$=\clj_{\bar{u}w}\; (\mbox{provided}\; wuw^*=\bar{u} )$$
\be 
=\beta_u\clj_w\;\mbox{by}\;(33)
\ee

\vsp 
Let $G$ be the simply connected Lie group $SU_2(\IC)$ and $g \raro u(g)$ be a $d-$dimensional unitary irreducible representation of $G$. Then there exists a $r \in U_d(\IC)$ such that 
\be 
r u(g) r^*=\bar{u(g)}
\ee  
for all $g \in G$. The element $r \in U_d(\IC)$ is determined uniquely modulo a phase factor in $S^1$. In particular any element
\be 
r_z=zr_0,\;z \in S^1
\ee 
satisfies (35), where we have fixed a $r_0$ satisfying (35) with additional condition 
\be 
r_0^2=I_d
\ee
In our notions $r_1=r_0$. In such a case, $r_{-1}=-r_0$ is the only other choice that satisfies (35) and (37) instead of $r_0$. 

\vsp 
With such a choice for $r_0$, for all $g \in SU_2(\IC)$ we have 
$$r_0u(g)r_0=\bar{u(g)}$$
Taking conjugation on both sides, we have 
$\bar{r_0} \bar{u(g)} \bar{r_0} = u(g)$
i.e. 
$$\bar{r_0}^* u(g) \bar{r_0}^* = \bar{u(g)}$$
Since $\bar{r_0}^2=I$, we arrive at 
$$\bar{r_0} u(g) \bar{r_0} = \bar{u(g)}$$
for all $g \in G$. So by the irreducible property of the representation $g \raro u(g)$, we conclude that $\bar{r}_0 = \mu r_0$, where $\mu^2=1$ since $r_0^2=\bar{r_0}^2=1$ i.e. $\mu$ is either $1$ or $-1$. 

\vsp 
Taking determinants of matrices on both sides of $\bar{r}_0r_0=\mu I_d$, we get $\mu^d=det(r_0)det(\bar{r}_0) = |det(r_0)|^2 =1$. This 
shows that $\mu=1$ if $d$ is an odd integer. For even values of $d$, we make a direct calculation to show $\mu=-1$ as follows: 

\vsp 
For $d=2$, let $\sigma_x,\sigma_y$ and $\sigma_z$ be the Pauli matrices
in $\!M_2(\IC)$ i.e. the standard ( irreducible ) representation of Lie algebra $su_2(\IC)$ in $\IC^2$:

\ben
\sigma_x = \left (\begin{array}{llll} 0&,&\; 1 \\ 1&,&\;\;0 \\ 
\end{array} \right ),
\een
\ben
\sigma_y = \left (\begin{array}{llll} \;0&,&\;\; i \\ -i&,&\;\;0 \\
\end{array} \right ),
\een
\ben
\sigma_z = \left (\begin{array}{llll} 1&,&\;\; 0 \\  0&,& \;-1\\
\end{array} \right ).
\een
A direct commutation shows that $r_0$ is given by  
\ben
r_0 = \left (\begin{array}{llll} \;0&,&\;\; i \\ -i&,&\;\;0 \\
\end{array} \right )
\een

\vsp 
The self-adjoint matrix $\sigma_y$ is also a  
unitary i.e. $\sigma_y^2=I_2$ and 
$$\sigma_y \sigma_x \sigma_y = - \sigma_x$$ 
and 
$$\sigma_y \sigma_z \sigma_y = -\sigma_z$$
Since $\sigma_x = \bar{\sigma_x}$ and $\sigma_z = \bar{\sigma_z}$, 
$\sigma_y$ inter-twins $e^{it\sigma_x}$ and $e^{it\sigma_z}$ with their conjugate matrices $e^{-it\sigma_x}$ and $e^{-it\sigma_z}$ respectively for all $t \in \IR$. In contrast, since $\bar{\sigma_y}=-\sigma_y$, we also get $\sigma_y$ inter-twins $e^{it\sigma_y}$ with $e^{-it\sigma_y}$ for all $t \in \IR$. So we set $r_0=\sigma_y$ ( other choice we can make for $r_0$ is $-\sigma_y$) and verify directly that $\bar{r_0}=-r_0$ i.e. $\mu=-1$ if $d=2$. 

\vsp 
We write $i\sigma_y=e^{it_0\sigma_y} \in SU_2(\IC)$, where $t_0={\pi \over 2}$ and verify that
$$u(e^{it_0\sigma_y})u(g)u(e^{-it_0\sigma_y})$$
$$=u(i\sigma_y) u(g) u(i\sigma_y)^*$$
$$=u((i\sigma_y) g (i\sigma_y)^*)$$
$$=u(\bar{g})$$
Since $su_2(\IC)$ is a real Lie algebra that has unique Lie algebra extension to a complex Lie algebra $sl_2(\IC)$, 
i.e. Lie algebra over the field of complex numbers, we also have 
$$u(\bar{g})=\bar{u(g)}$$ 
for all $g \in SU_2(\IC)$ ( Lie-derivatives of the representations in both sides are equal as element in $sl_2(\IC)$). 
So we have  
\be 
u(e^{it_0\sigma_y})u(g)u(e^{-it_0\sigma_y})
=\overline{u(g)}
\ee
If $\pi_u$ is the associated Lie-representation of $su_2(\IC)$, we have 
$$u(e^{it_0\sigma_y})=e^{it_0\pi_{u}(\sigma_y)}$$
for even integer values of $d$, whereas 
$$u(e^{it_0\sigma_y})=e^{2it_0\pi_{u}(\sigma_y)}$$ 
for odd integer values of $d$. Thus for an arbitrary even values of $d$, the unitary matrix $r_0 = e^{it_0\pi_{u}(\sigma_y)})$ satisfies (35) and (37). In contrast, for an arbitrary odd values of $d$, the unitary matrix $r_0=e^{i2t_0\pi_u(\sigma_y)}$ satisfies (35) and (37). In short, $\mu=1$ if $d$ is an odd integer and $-1$ if $d$ is an even integer. 

\vsp 
We write $\mu=\zeta^2$ and set $r_0 \in U_d(\IC)$, such that 
$$\zeta r_0 = u(e^{it_0\sigma_y}) \in U_d(\IC),$$  
where $\zeta^2=\mu$ and so $\mu$ is $1$ for odd values of $d$ otherwise $-1$. 
In the last section, we will recall standard explicit description of $r_0$ and $g \raro u(g)$ 
that satisfies (35) and (37). Note also that $r_{\zeta}=\zeta r_0$ is a matrix with real entries 
irrespective of values taken for $d$. 

\vsp 
Now we go back to our main text. So we have 
\be 
\clj_{r_z}^2(x)=\beta_{\bar{r_z}r_z}(x)=\beta_{\mu I_d}(x)
\ee 
for all $x \in \tilde{\clo}_d \otimes \clo_d$, where $\mu=1$ or $-1$ depending on $d$ odd or even. In any case, by 
(34) and (35), we also have  
\be 
\clj_{r_z} \beta_{u(g)} = \beta_{u(g)} \clj_{r_z}
\ee
for all $g \in SU_2(\IC)$.  

\vsp 
Let $\omega$ be a translation invariant factor state of $\IM$ and $\psi$ be an extremal element in $K_{\omega}$. We define a state $\psi_0: \tilde{\clo}_d \otimes \clo_d \raro \IC$ by extending both $\tilde{\psi}: \tilde{\clo}_d \raro \IC$ and $\psi:\clo_d \raro \IC$ by 
\be 
\psi_0(\tilde{s}_{I'}\tilde{s}^*_{J'} \otimes s_Is_J^*) 
= < \zeta_{\psi}, \tilde{v}_{I'}\tilde{v}^*_{J'}v_I^*v_J^* \zeta_{\psi}>
\ee
for all $|I'|,|J'|,|I|$ and $|J| < \infty$. Proposition 2.4 says that $(\tilde{\clh} \otimes_{\clk} \clh, \pi, \zeta_{\psi})$ is the GNS representation $(\clh_{\psi_0},\pi_{\psi_0},\zeta_{\psi_0})$ of $(\tilde{\clo}_d \otimes \clo_d,\psi_0)$. 

\vsp 
\begin{pro} 
Let $\omega$ be an extremal point in the convex set of translation invariant states of $\IM$. If $\omega$ is also 
$SU_2(\IC)$ invariant then the following statement holds for any extremal elememt $\psi \in K_{\omega}$:

\vsp 
\NI (a) $\psi \beta_{u(g)} = \psi$ on $\clo_d$;

\vsp 
\NI (b) $\psi_0 \beta_{u(g)} \otimes \beta_{u(g)} = \psi_0$ on $\tilde{\clo}_d \otimes\clo_d$ for all 
$g \in SU_2(\IC)$; 

\vsp 
\NI (c) $\psi_0 \beta_{r_{\zeta}} \otimes \beta_{r_{\zeta}} = \psi_0$ on $\tilde{\clo}_d \otimes\clo_d$, where $r_{\zeta} = \zeta r_0 \in u(SU_2(\IC))$;

\vsp 
\NI (d) $\psi_0 \beta_{r_0} \otimes \beta_{r_0} = \psi_0$ on $\tilde{\mbox{UHF}}_d \otimes \mbox{UHF}_d$.

\end{pro} 

\vsp 
\begin{proof}
For (a),(b) and (c) we refer to Proposition 3.1 in \cite{[24]}, where we used factor property of $\psi$ but same holds good for ergodic state as well once we use Proposition 2.1 (a). The last statement is a simple consequence of (c) since 
$\beta_{r_{\zeta}} = \beta_{r_0}$ on $\mbox{UHF}_d$ and $\tilde{\mbox{UHF}_d}$ as $r_{\zeta} = \zeta r_0$, where $\zeta^2
=\mu$ is either $1$ or $-1$. 
\end{proof}

\section{ Real, lattice reflection symmetric with a twist, $SU_2(\IC)$ and translation invariant ergodic states }

\vsp 
We quickly recall from \cite{[23]} the folloing definitions. Given a $\lambda$-invariant state of $\clo_d$, we define state $\tilde{\psi}:\clo_d \raro \IC$ by
$$\tilde{\psi}(s_Is_J^*)=\psi(s_{\tilde{I}}s^*_{\tilde{J}})$$
for all $|I|,|J| < \infty$ and extend linearly. Also we consider the state $\bar{\psi}:\clo_d \raro \IC$ defined by 
$$\bar{\psi}(s_Is_J^*)=\psi(s_Js^*_I)$$
for all $|I|,|J| < \infty$ and extend linearly. So $\bar{\psi}$ and $\tilde{\psi}$ are well defined $\lambda$-invariant states on $\clo_d$.  

\vsp 
Let $S_{\theta,\IZ_2}$ be the convex subset of translation invariant defined by  
$$S_{\theta,\IZ_2} = \{ \omega: \omega(Q) = \omega \theta(Q),\;\omega(Q)=\omega(\clj_{r_{\zeta}}(Q^*)),\forall Q \in \IM \}$$
We recall from (10) that $\clj_{r_{\zeta}}(Q)=\overline{\beta_{r_{\zeta}}(\tilde{Q})}$ and so 
$\clj_{r_{\zeta}}(Q^*)=\beta_{r_{\zeta}}(\tilde{Q^t})$. So the map $Q \raro \tilde{\beta}_{r_{\zeta}}(Q)=\clj_{r_{\zeta}}(Q^*)=\beta_{r_{\zeta}}(\tilde{Q^t})$ is linear and anti-automorphism on $\IM$. It is obvious that any translation invariant real and lattice reflection symmetric state $\omega$ with twist $\beta_{r_{\zeta}}$ is an element in $S_{\theta,\IZ_2}$. 

\vsp 
If $\omega$ is an extremal element in $S_{\theta,\IZ_2}$ then there exists an extremal translation invariant state $\omega'$ of $\IM$ such that 
$$2\omega = \omega' + \omega' \tilde{\beta}_{r_{\zeta}};$$ 
For a proof, we use extremal decomposition of $\omega$ in the comapct convex 
set of translation invariant state of $\IM$ and use the fact that $\tilde{\beta}_{r_{\zeta}}^2=I$ on $\IM$.

\vsp 
For any element $\omega \in S_{\theta,\IZ_2}$, we also consider the set 
$$K_{\omega} = \{ \psi:\clo_d \raro \IC \;\mbox{state},\; \psi = \psi \lambda,\; \psi_{|}\mbox{UHF}_d=\omega|\IM_R \}$$
as in section 2. We also consider 
$$K_{\omega,\IZ_2} = \{\psi \in K_{\omega}: \psi_0(\clj_{r_{\zeta}}(x^*)) = \psi_0(x),\; x \in \tilde{\clo}_d \otimes \clo_d  \}$$ 
The set 
$K_{\omega,\IZ_2}$ is a non empty compact convex subset of $K_{\omega}$ since 
$$\psi = {1 \over 4}(\sum_{0 \le k \le 3} \psi' \tilde{\beta}^k_{r_{\zeta}})$$ 
is an element in $K_{\omega,\IZ_2}$ for any element $\psi' \in K_{\omega'}$, where 
$\tilde{\beta}_{r_{\zeta}}(x) = \clj_{r_{\zeta}}(x^*)$ is a linear anti-automorphism on 
$\tilde{\clo}_d \otimes \clo_d$ and $\tilde{\beta}^4_{r_{\zeta}}=I$. However, since 
$$\beta_z \clj_{r_{\zeta}} = \clj_{r_{\zeta}} \beta_{\bar{z}}$$ 
for all $z \in S^1$, we can not claim that $\psi \beta_z \in K_{\omega,\IZ_2}$ whenever 
$\psi \in K_{\omega,\IZ_2}$ unless $\psi \beta_{z^2} = \psi$.  

\vsp 
\begin{pro} the following statements hold:

\vsp 
\NI (a) Let $\omega$ be an element in the non-empty convex compact set $S_{\theta,\IZ_2}$.
For an element $\psi$ in $K_{\omega,\IZ_2}$, the associated amalgamated representation $\pi_{\psi_0}:\tilde{\clo}_d \otimes \clo_d \raro \clb(\tilde{\clh}\otimes_{\clk}\clh)$ determines an anti-automorphism 
$\clj_{r_{\zeta \zeta_0}}$ that takes 
$$\pi_{\psi_0}(\tilde{s}_{I'}\tilde{s}_{J'}s_Is^*_J) \raro\psi_{\psi_0}(\beta_{r_{\zeta}}(\tilde{s}_I\tilde{s}^*_Js_{I'}s_{J'}^*))$$ 
extending anti-linearly such that
\be 
\clj_{r_{\zeta}}(X) = \clj_{r_{\zeta}} X \clj^*_{r_{\zeta}}
\ee 
where $\clj_{r_{\zeta}}$ is an anti-unitary operator on $\clh \otimes_{\clk} \tilde{\clh}$ that takes
$$ 
\pi(s_Is^*_J\tilde{s}_{I'}\tilde{s}^*_{J'})\zeta_{\psi} 
\raro \beta_{r_{\zeta}}(\pi(s_{I'}s^*_{J'}\tilde{s}_{I}\tilde{s}^*_{J}))\zeta_{\psi}
$$
for all $\;|I|,|J|,|I'|,|J'| < \infty $ and then extending anti-linearly on their linear span.

\vsp 
\NI (b) If $\omega \in S_{\theta,\IZ_2}$ and $\bar{\omega}(x) = \omega \clj_{I_d}(x^*)$ i.e.  
lattice reflection symmetric state of $\IM$ satisfying $\omega = \tilde{\omega}$ then 
$\bar{\omega} = \omega \beta_{r_{\zeta}}$. 

\vsp 
\NI (c) If $\bar{\omega}=\tilde{\omega} \beta_{r_{\zeta}}=\omega$ then $\omega \in S_{\theta,\IZ_2}$. Futhermore, if $\omega$ is also an extremal element in the convex set of translation invariant states of $\IM$ then there exists an extremal element $\psi \in K_{\omega}$ such that $\bar{\psi}=\tilde{\psi} \beta_{r_{\zeta}}=\psi \beta_{\zeta_0}$, where $\zeta^2\zeta_0^2 \in H$. Furthermore, $\psi$ is also an extremal element in $K_{\omega,\IZ_2}$.

\vsp 
\NI (d) If $\omega$ in (c) is also pure then the Popescu elements $(\clk,\clm,v_k,1 \le k \le d,\phi)$ of $\psi$ as given in Proposition 2.4 satisfies the following:

\vsp 
\NI (i) There exists a unique unitary operator $\gamma_{r_{\zeta}}$ on $\clk$ such that $\gamma_{r_{\zeta}} \zeta_{\psi}=\zeta_{\psi}$ and 
\be 
\gamma_{r_{\zeta}} ( \sum c_{I'J',I,J} \tilde{v}_{I'}\tilde{v}_{J'}^*v_Iv^*_J) )\gamma^*_{r_\zeta} = \sum c_{I',J',I,J} \clj \beta_{r_{\zeta}}(\tilde{v}_{I}\tilde{v}_{J}v_{I'}v_{J'}^*)\clj
\ee
for all $|I'|,|J'|,|I|$ and $|J| < \infty$, where $\gamma_{r_{\zeta}}$ is
commuting with modular elements $\Delta^{1 \over 2}, \clj$. 

\vsp 
\NI (ii) $\gamma_{r_{\zeta}} u_z = u_{\bar{z}} \gamma_{r_{\zeta}}$ for all $z \in H$.  

\vsp 
\NI (iii) The anti-unitary map $\clj \gamma_{r_{\zeta}}: \clk  \raro  \clk$ extends to the 
anti-unitary map $\clj_{r_{\zeta}}$ on $\tilde{\clh} \otimes_{\clk} \clh$.   

\vsp 
\NI (iv) $Ad_{\gamma_{r_{\zeta}}^2}=\beta_{\zeta^2 I_d}$ and $\gamma_{r_{\zeta}}$ is self-adjoint if and only if $\zeta^2=1$.

\end{pro} 

\vsp 
\begin{proof} 
For (a), as a first step, we verify the following identities:
$$\langle \beta_{r_{\zeta}}(S_{I'}S^*_{J'}\tilde{S}_I\tilde{S}^*_J) \zeta_{\psi},  \zeta_{\psi} \rangle$$
$$=\overline{\psi(\beta_{r_{\zeta}}(s_{I'}s^*_{J'}\tilde{s}_I\tilde{s}^*_J))}$$
$$=\overline{\psi (\clj_{r_{\zeta}}(\tilde{s}_{I'}\tilde{s}^*_{J'}s_Is^*_J))}$$ 
$$=\psi(\tilde{s}_{I'}\tilde{s}^*_{J'}s_Is^*_J)$$
$$=\langle \zeta_{\psi},\tilde{S}_{I'}\tilde{S}^*_{J'}S_IS^*_J \zeta_{\psi} \rangle$$
For more general elements, we use Cuntz relation (16) as in Theorem 3.5 of \cite{[25]} 
to prove that the map $\clj_{r_{\zeta}}$ is indeed an anti-unitary map. 

\vsp 
Since $\clj_{r_{\zeta}} = \clj_{I_d} \beta_{r_{\zeta}}$ by (30), (b) is obvious. 

\vsp 
For (c), we closely follow the argument used in Proposition 3.4 of \cite{[25]}. Here, we quickly repeat the argument used in the proof for Proposition 3.4 (b) in \cite{[25]}, where we had used argument for real reflection symmetry state without twist. We fix any extremal elememt $\psi \in K_{\omega}$ and verify that
$\tilde{\psi} \beta_{r_{\zeta}} \in K_{\omega}$ since its restriction to $\mbox{UHF}_d$ is $\tilde{\omega} \beta_{r_{\zeta}}$ is equal to $\omega$. So
there exists $\zeta_0 \in S^1$ such that $\tilde{\psi} \beta_{r_{\zeta}} = \psi \beta_{\zeta_0}$. Since $\beta_z$ commutes with $\beta_{r_{\zeta}}$ and $\tilde{\psi \beta_z}=\tilde{\psi} \beta_z$ for all $z \in S^1$, we cnclude that
$\tilde{\psi} \beta_{r_{\zeta}} = \psi \beta_{\zeta_0}$ for all
$\psi \in K_{\omega}$. That $\zeta^2\zeta_0^2 \in H$ follows from
$\beta_{r_{\zeta}}^2=\beta_{\zeta^2}$.

\vsp
We fix an extremal element $\psi' \in K_{\omega}$. By (b), there exists $z_0 \in S^1$ such that $\bar{\psi'}= \psi' \beta_{z_0}$. Since $\bar{\psi' \beta_z}=\bar{\psi'} \beta_{\bar{z}}$ for all $z \in S^1$, the affine map $\psi \raro \bar{\psi}$ takes $\psi' \beta_z$ to
$\psi' \beta_{z_0\bar{z}}=\psi' \beta_z \beta_{z_0\bar{z}^2}$.
We choose $\psi=\psi' \beta_z$ with $z$ satisfying $z_0\bar{z}^2=\zeta_0$.

\vsp 
Thus there exists an extremal element $\psi$ in $K_{\omega}$ satisfying $\tilde{\psi} \beta_{r_{\zeta}}=\bar{\psi}=\psi \beta_{\zeta_0}$, where $\zeta_0\zeta$ is either $1$ or $\zeta_0\zeta \in \{1, {i\pi \over n} \}$
for $H=\{z \in S^1: z^n=1 \}$.

Now we also verify
$$\psi \clj_{r_{\zeta}}((s_Is_J^*)^*)$$
$$=\psi \tilde{\beta}_{r_{\zeta}}(s_Js_I^*)$$
$$=\psi \beta_{\zeta_0}(s_Js_I^*)$$
$$=\bar{\psi}(s_Js_I^*)$$
$$=\psi(s_Is_J^*)$$
Along the same line we can verity $\psi \clj_{r_{\zeta}}(x^*)=\psi(x)$ for all $x=s_Is^*_Js_{\bar{I}}s^*_{\bar{J}}$ and then extend for any $x$ of their linear sums. 

\vsp 
Since $K_{\omega,\IZ_2}$ is a convex subset of $K_{\omega}$, any extremal element in $K_{\omega}$ that is also an element in $K_{\omega,\IZ_2}$, is also extremal in $K_{\omega,\IZ_2}$.

\vsp 
Rest of the proof is given in Theorem 3.5 in \cite{[23]} once we take $r_{\zeta}=g_0$ as
$r_{\zeta}$ is a matrix with real entries. For details we refer to Theorem 3.5 in \cite{[25]}.  

\end{proof} 

\vsp 
In general ergodic states of $(\IM,\theta)$ need not be factor states of $\IM$ \cite{[8]} though ergodic states of $(\clo_d,\lambda)$ are factor states of $\clo_d$. One of the central question that arises while dealing with ground states of Hamiltonian $H_{XXX}$, whether additional symmetries of $H$ make extremal decomposition of its ground states $\omega$ in the convex set of translation and $G$-invariant states to be a factor decomposition? In the following $G$ is a compact group and $\psi \raro \psi \beta_g$ is a $G$-action on states of $\clo_d$ commuting with the action $\psi \raro \psi \lambda$.

\vsp 
\begin{pro} 
Let $\omega$ be an extreme point in the convex set of translation and $G$-invariant states of $\IM$ and $K_{\omega,G}$ be the non-empty compact convex set of $\lambda$ and $G$-invariant states $\psi$ so that $\psi|\mbox{UHF}_d=\omega|\IM_R$ as in Proposition 2.1. Then the following statements are true:

\vsp 
\NI (a) The set $K_{\omega,G}$ is a face in the convex set of translation and $G$ invariant states of $\clo_d$;

\vsp 
\NI (b) For an extremal element $\psi \in K_{\omega,G}$, let $\psi = \int^{\oplus}\psi'_{\alpha}d\mu(\alpha)$ be an ergodic decomposition in the convex set of $\lambda$-invariant states of $\clo_d$ and $\psi_{\alpha}= \int_{G}\psi'_{\alpha} 
\beta_{g} dg$, where $dg$ is the normalised Haar measure on $G$. Then $\mu$-almost every where, $\psi = \psi_{\alpha}$ i.e. 
$$\psi = \int_G \psi'\beta_{g}dg$$
for some extremal $\psi'$ in the convex set of $\lambda$ invariant states.  

\vsp 
\NI (c) For two elements $g_1$ and $g_2$ in $U_d(\IC)$, either $\psi'\beta_{g_1}$ and $\psi' \beta_{g_2}$ are same or orthogonal and 
$$\pi_{\psi}(x) = \int^{\oplus}_{G / G'} \pi_{\psi'\beta_{g'}}(x)dg'$$
for all $x \in \clo_d$, where $G'=\{g \in G: \psi' \beta_{g} = \psi' \}$ and $dg'$ is the induced measure on the cosets $G / G'$ of $G'$. Moreover, $L^{\infty}(G / G',dg') \otimes I \subset \pi_{\psi}(\clo_d)''$.   

\vsp 
\NI (d) If $\psi_1$ and $\psi_2$ are two extreme points in $K_{\omega,G}$ then $\psi_2=\psi_1 \beta_z$ for some $z \in S^1$ provided $G$ action $g \raro \beta_g$ commutes with $(\beta_z:z \in S^1)$. In such a case, the closed set $H = \{ z \in S^1: \psi \beta_z = \psi \}$ is independent of the extreme point $\psi \in K_{\omega,G}$; 

\vsp 
\NI (e) If $\omega$ is an extremal element in the convex set of translation, reflection symmetric and $G$-invariant states of $M$ then statements (a)-(d) are also valid for 
$K_{\omega,G} \bigcap \{\psi: \tilde{\psi} = \psi \}$

\end{pro}

\vsp 
\begin{proof} 
For any two $G$ and $\lambda$-invariant states $\psi$ and $\psi'$ and $\lambda \in (0,1)$, if the state $\psi_{\lambda}=\lambda \psi_1 +(1-\lambda)\psi_0$ is in $K_{\omega,G}$ then its restriction to $\mbox{UHF}_d$ being an extremal element in the convex set of $G$ and $\lambda$ invariant states, $\psi_1=\psi_2=\omega$ on $\mbox{UHF}_d$. So $\psi_1$ and $\psi_1$ are elements in $K_{\omega,G}$. This proves (a).   

\vsp 
The statement (b) follows as $\psi_{\alpha} \in K_{\omega,G}$,
$$\int \psi_{\alpha}d\mu(\alpha) $$
$$= \int \psi'_{\alpha}\beta_{g}dg d\mu(\alpha)$$
$$= \int \psi'_{\alpha} d\mu(\alpha) \beta_{g} dg$$
$$= \int \psi \beta_g dg$$  
$$=\psi$$
and $\psi$ is an extremal element in $K_{\omega,G}$. 

\vsp 
For (c), we recall by Proposition 2.1, an extremal element $\psi'$ in $K_{\omega'}$ is a factor state of $\clo_d$. So by Proposition 2.4.47 in \cite{[7]}, two factor states $\psi' \beta_{g_1}$ and $\psi' \beta_{g_2}$ are quasi-equivalent if anly only if ${1 \over 2}(\psi' \beta_{g_1} + \psi' \beta_{g_2})$ is a factor state. However, any translation invariant factor state is also an extremal element in $K_{\omega}$ and so $\psi'\beta_{g_1}$ and $\psi' \beta_{g_2}$ are quasi-equivalent if and only if they are equal. Since $\psi' \beta_{g_1}$ and $\psi' \beta_{g_2}$ are trivially centrally ergodic by $\lambda$, (for details we refer to Lemma 7.4 in \cite{[9]} ), we conclude as in Lemma 7.4 that $\psi' \beta_{g_1}$ and $\psi' \beta_{g_2}$ are either same or disjoint by Theorem 4.3.19 in \cite{[7]}, where 
$(\beta_{g}:g \in G)$ commutes with $\lambda$.  

\vsp 
For the last statement $L^{\infty}(G / G_{\alpha},dg') \otimes I \subset \pi_{\psi}(\clo_d)''$, $\psi'$ being an extremal element in the $\lambda$ invariant states of $\clo_d$, we have 
$${1 \over n} \sum_{0 \le k \le n-1}\pi_{\psi}(\lambda^k(x))$$
$$= \int^{\oplus}_{G \ G'} {1 \over n} \sum_{0 \le k \le n-1}\pi_{\psi'_{\alpha}}(\lambda^k(\beta_{g'}(x))dg'$$
$$\raro \int^{\oplus}_{G \ G'} \psi'_{\alpha}(\beta_{g'}(x))dg'$$
So $g' \raro \psi'_{\alpha}(\beta_{u(g')}(x))$ is in $\pi_{\psi}(\clo_d)''$. 
Since the collection of functions seperates points in $G / G_{\alpha}$, we conclude that $L^{\infty}(G / G_{\alpha},dg') \otimes I_{\alpha'} \subset \pi_{\psi}(\clo_d)''$. This shows that the extremal decomposition is a central decomposition. 

\vsp 
For (d), we use decomposition given in (b) for $\psi_1$ and $\psi_2$ in the convex set of $\lambda$-invariant states for 
$$\psi_k = \int_G \psi'_k \beta_{g}dg$$
and use Lemma 7.4 in \cite{[8]} for
$\psi'_1 = \psi'_2 \beta_z$ for some $z \in S^1$. 
Since $\beta_z$ commutes with $\{\beta_{g}:g \in G \}$, we conclude that 
$\psi_1 = \psi_2 \beta_z$ as well by (b).  

\vsp 
For (e), we verify that $\tilde{\psi \beta_g}=\tilde{\psi} \beta_g$ and thus the action 
$\psi \raro \psi \beta_g$ commutes with $\psi \raro \tilde{\psi}$. So we may repeat arguments used for (a) to (d). 
\end{proof} 

\vsp 
The map $x \raro \tilde{\beta}_{r_{\zeta}}(x) = \clj_{r_{\zeta}}(x^*)$ is a linear but anti 
automorphism on $\tilde{\clo}_d \otimes \clo_d$. It is a $\IZ_2$-action on $\tilde{\mbox{UHF}}_d\otimes \mbox{UHF}_d$ and it extends to a $\IZ_2$ action on $\tilde{\clo}_d \otimes \clo_d$ if and only if $\zeta^2=1$. In the following proposition, we use commuting property (40) of the $\IZ_2$-action $\psi \raro \psi \tilde{\beta}_{r_{\zeta}}$ with $SU_2(\IC)$-action $\{ \psi \raro \psi \beta_{u(g)}:g \in SU_2(\IC) \}$ for a natural group $G=SU_2(\IC) \otimes \IZ_2$ action extension. 

\vsp 
We consider the following non-empty compact convex sets
$$S_{\theta,G} = \{ \omega \in S_{\theta,\IZ_2}: \omega = \omega \beta_{u(g)},\;\forall g \in SU_2(\IC) \}$$
and 
$$K_{\omega,G} = \{ \psi \in K_{\omega,\IZ_2}: \psi=\psi \beta_{u(g)},\;\forall \;g \in SU_2(\IC) \}$$
for $\omega \in S_{\theta,G}$. So $S_{\omega,G} \subseteq S_{\omega,\IZ_2}$. 

\vsp 
\begin{pro} 
Let $\omega \in S_{\theta,G}$ and $\psi \in K_{\omega,G}$. We consider the 
anti-automorphism $\clj_{r_{\zeta}}$ on $\pi_{\psi_0}(\tilde{\clo}_d \otimes \clo_d)''$ 
that takes 
$$\pi_{\psi_0}(\tilde{s}_{I'}\tilde{s}^*_{J'}s_Is_J^*) \raro \pi_{\psi_0}(\beta_{r_{\zeta}}(\tilde{s}_I\tilde{s}_J^*s_{I'}s_{J'}^*))$$
by extending anti-linearly in their linear span defined as in Propsotion 3.1.  

\vsp  
If $\omega$ is an extremal point in $S_{\theta,G}$ and $\psi$ is an extremal element in $K_{\omega,G}$ then the following statement are true:

\vsp 
\NI (a) There exists an extremal element $\psi' \in K_{\omega}$ such that 
$$\psi = \int^{\oplus}_{G / G'} \psi' \beta_{g'} dg'$$ 
is a factor decomposition, where $dg$ is the normalised Haar measure on $G$ and $dg'$ is the induced normalised measure on the quotient space $G/ G'$, where $G'=\{g \in G: \psi' \beta_{g} = \psi' \}$.

\vsp 
\NI (b) There exists an extremal element $\psi'' \in K_{\omega,\IZ_2}$ such that 
$$\psi = \int^{\oplus}_{SU_2(\IC) / SU_2(\IC)''} \psi'' \beta_{u(h'')}dh''$$ 
is an extremal decomposition of $\psi$ in the convex set $K_{\omega,\IZ_2}$, 
where 
$SU_2(\IC)''=\{ h \in SU_2(\IC): \psi'' \circ \beta_{u(h)} =\psi'' \}$ and 
$dh''$ is the induced probability measure on the quotient space $SU_2(\IC)/ SU_2(\IC)''$. 

\vsp 
\NI (c) There is a choice for an extremal element $\psi'$ in $K_{\omega'}$, 
where $\omega'=\psi'$ on $\mbox{UHF}_d = \IM_R$. in the factor decomposition given 
in (a) such that the following statements hold: 

\vsp 
\NI (c1) If $d$ is an odd integer then     
$$2\psi'' = \psi' + \psi' \tilde{\beta}_{r_{\zeta}}$$
on $\tilde{\clo}_d \otimes \clo_d$, 

\vsp 
\NI (c2) If $d$ is an even integer then 
$$4\psi'' = \sum_{0 \le k \le 3} \psi' \tilde{\beta}^k_{r_{\zeta}}$$
on $\tilde{\clo}_d \otimes \clo_d$, where $\omega'_{|}\IM_R = \psi'_{|}\mbox{UHF}_d$. 

\vsp 
\NI (d) There exists an anti-automorphism $\clj''_{r_\zeta}$ on $\pi_{\psi_0''}(\tilde{\clo}_d \otimes \clo_d)''$ satisfying 
$$\clj_{r_{\zeta}} = \int_{SU_2(\IC) / SU_2(\IC)''}^{\oplus}\clj''_{r_{\zeta}} \beta_{u(g')} dg'$$
on $\tilde{\clo}_d \otimes \clo_d$  and 
$SU_2(\IC)'=\{ g \in SU_2(\IC): \psi' \beta_{u(g)} = \psi' \}$. 
In such a case, 
$$\psi''_0 \clj''_{r_{\zeta}}(X^*)=\psi''_0(X)$$ 
for all $X \in \pi_{\psi_0''}(\tilde{\clo}_d \otimes \clo_d)''$;

\vsp 
\NI (e) There exists an extremal element $\omega''$ in the compact convex set 
$$S_{\theta,\IZ_2}=\{ \omega : \omega \theta (x)= \omega(x), \; \omega(x) = \omega \clj_{r_{\zeta}}(x^*)\;\forall x \in \tilde{\mbox{UHF}}_d \otimes \mbox{UHF}_d  \}$$ 
such that $$\omega = \int^{\oplus}_{SU_2(\IC) / SU_2(\IC)'} \omega'' \beta_{u(h')}dh'$$ 
is an extremal decomposition of $\omega$ in the convex set $S_{\omega,\IZ_2}$ and 
$$\omega''(x) = {1 \over 2}(\omega'(x) + \omega'\clj_{r_{\zeta}}(x^*))$$
for all $x \in \tilde{\mbox{UHF}}_d \otimes \mbox{UHF}_d$, is a choice for an extremal 
element $\omega'$ in the convex set of translation invariant states of $\IM$, where 
$SU_2(\IC)'=\{ h \in SU_2(\IC): \omega' \circ \beta_{u(h)} =\omega' \}$ and 
$dh'$ is the induced probability measure on the quotient space $SU_2(\IC)/ SU_2(\IC)'$.  

\vsp 
\NI (f) There exists an anti-automorphism $\clj''_{r_{\zeta}}$ on $\pi_{\omega''}(\tilde{\mbox{UHF}}_d \otimes \mbox{UHF}_d)''$ satisfying 
$$\clj_{r_{\zeta}} = \int_{SU_2(\IC) / SU_2(\IC)'}^{\oplus}\clj''_{r_{\zeta}} \beta_{u(g')} dg'$$
on $\pi_{\omega}(\mbox{UHF}_d\otimes \mbox{UHF}_d)''$, where  
$$\omega''(\clj''_{r_{\zeta}}(X^*))=\omega''(X)$$ 
for all $X \in \pi_{\psi_0''}(\tilde{\mbox{UHF}}_d \otimes \mbox{UHF}_d)''$; 

\end{pro} 

\vsp  
\begin{proof} 
It is a rouite work to check by Proposition 3.2 (a), (b) and (c) that statements (a) and (b) are valid for any extremal element $\psi \in K_{\omega,G}$ once we take its extremal decomposition in the convex sets $K_{\omega}$ and $K_{\omega,\IZ_2}$ repectively. For (c), we consider extremal decomposition of $\psi''$ in the convex set of $\lambda$-invariant states of $\clo_d$ and use $\clj_{r_{\zeta}}^2=\beta_{\mu I_d}$, where $\mu=1$ for odd values of $d$ and $\mu=-1$ for even values of $d$. For (d), we use (b) and Proposition 3.1 (a) and commuting property $\clj_{r_{\zeta}}$ with $\{\beta_{u(g)}:g \in SU_2(\IC) \}$ given in (40). The last two statements (e) and (f) are essentially re-statements of (c) and (d) respectively once restricted to $\tilde{\mbox{UHF}}_d \otimes \mbox{UHF}_d \subset \tilde{\clo}_d \otimes \clo_d$.     
\end{proof} 

\vsp 
Alternatively, we may write 
$$S_{\theta,\IZ_2}= \{ \omega \in S_{\theta}: \bar{\omega} = \tilde{\omega} \circ\beta_{r_{\zeta}} \}$$
and consider the convex subset of $S_{\theta}$
$$S_{\theta,\IZ_2,+} = \{ \omega \in S_{\theta}: \omega(\clj_{r_{\zeta}}(x)x) \ge 0,\;\;\forall x \in \IM \}$$

\vsp 
\begin{lem} The following statements hold:

\vsp 
\NI (a) $S_{\theta,\IZ_2,+} \subset S_{\theta,\IZ_2}$; 

\vsp 
\NI (b) If $\omega \in S_{\theta,\IZ_2}$ then $\omega \beta_{u(g)} \in S_{\theta,\IZ_2}$ for
all $g \in SU_2(\IC)$;

\vsp 
\NI (c) If $\omega \in S_{\theta,\IZ_2,+}$ then $\omega \beta_{u(g)} \in S_{\theta,\IZ_2,+}$ 
for all $g \in SU_2(\IC)$.

\end{lem}

\vsp 
\begin{proof} 
We have already discussed a proof in the introduction. We include now a formal proof
for (a). We set sesqui-linear map  
$$(x,y) = \omega(\clj_{r_{\zeta}}(x)y)$$ 
on $\clm \times \clm$ and verify by sesqui-linear property that  
$$(x,y) = \sum_{0 \le k \le 3} (x+i^ky,x+i^ky)$$
Since $(x,x)$ are real numbers for all $x \in \IM$, we verify directly that
$(x,y)=\overline{(y,x)}$ for all $x,y \in \IM$. By taking $y=1$ in the relaion, 
we conclude $x \in S_{\theta,\IZ_2}$.  

\vsp 
We use the commuting property $\clj_{r_{\zeta}} \beta_{u(g)}=\beta_{u(g)} \clj_{r_{\zeta}}$ for all $g \in SU_2(\IC)$ to prove (b) and (c). 
To that end we verify the following equalities for (b):
$$\omega \beta_{u(g)} \clj_{r_{\zeta}}(x^*)$$ 
$$=\omega \clj_{r_{\zeta}} \beta_{u(g)}(x^*)$$
$$=\omega \beta_{u(g)}(x)$$
for $\omega \in S_{\theta,\IZ_2}$.

\vsp 
For (c), we also verify that
$$\omega(\beta_{u(g)}(\clj_{r_{\zeta}}(x))x)$$
$$=\omega(\clj_{r_{\zeta}}(\beta_{u(g)}(x))\beta_{u(g)}(x)) \ge 0$$ 
if $\omega \in S_{\theta,\IZ_2,+}$. 

\end{proof} 

\vsp 
\begin{pro} 
Let $\omega$ be an extremal element in the convex set $S_{\theta,G}$ and $\psi$ be an extremal element in $K_{\omega,G}$ as in Proposition 3.3. We consider direct integral representation 
$$\omega = \int^{\oplus}_{SU_2(\IC) / SU_2(C)'} \omega'' \beta_{u(h')}dh'$$
given in Proposition 3.3 (e) in the convex set $S_{\theta,\IZ_2}$. Then the following statements are true:

\vsp
\NI (a) $\omega''=\omega'' \beta_{u(g)}$ for all $g \in SU_2(\IC)$;

\vsp
\NI (b) $\omega$ is also an extremal element in $S_{\theta,\IZ_2}$;

\vsp
\NI (c) Let $\omega$ be an element in $S_{\theta,G}$ and $\omega=\int \omega_{\alpha}d\mu(\mu)$ be an extremal decomposition in the convex set
$S_{\theta,G}$ then $\omega_{\alpha}$ are also extremal elememts in $S_{\theta,\IZ_2}$.

\end{pro}

\vsp
\begin{proof}
By Proposition 3.3 (e), for each $g \in SU_2(\IC)$, the states $\omega''\beta_{u(g)}$ and $\omega'' \beta_{u(g)} \beta_{r_{\zeta}}$ are either orthogonal to each other or equal. Suppose these two states are orthogonal to each other for each $g \in SU_2(\IC)$. Then $\omega$ is orthogonal to $\omega \beta_{u(g)} \beta_{r_{\zeta}}\beta_{u(g)^*}=\omega \beta_{u(ghg^{-1})}$ for all $g \in SU_2(\IC)$, where  $r_{\zeta}=u(h)$ is as described in (38). We recall that $r_{\zeta}=\zeta r_0 = u(h)$, where $h=e^{\pi i\sigma_{\it y}} \in SU_2(\IC)$ is given explicitely in section 2, where $\zeta^2=1$ for odd integer values of $d$ and $\zeta^2=-1$ for even integer values of $d$.

\vsp
The normal subgroup generated by $\{ghg^{-1}: g \in SU_2(\IC) \}$ is not equal to the subgroup $\{-I,I\}$ as $h \notin \{-I,I\}$. Thus the generated normal subgroup is equal to the entire group $SU_2(\IC)$. Thus $\omega''$ is orthogonal to $\omega'' \beta_{u(g)}$ on $\IM$ for $g \in SU_2(\IC)$. This brings a contradiction since $\omega''$ can not orthogonal to itself. So the set $\{g \in SU_2(\IC): \omega'' \beta_{u(g)}=\omega'' \beta_{u(g)}\beta_{r_{\zeta}}$ is a non empty and a closed subset of $SU_2(\IC)$.

\vsp
For a given state $\omega$ of a $C^*$ algebra $\cla$, the collectioon of states $\Omega^{\perp}=\{\omega': \omega' \perp \omega \}$ is closed set in the weak$^*$ topology of $\cla^*$. For a proof we consider universal representation $(\clh,\pi)$ of $\cla$ \cite{[Tak]} and for each state $\omega$, there is a unique
projection $E_{\omega} \in \pi(\cla)'$ such that $\pi_{\omega}(x)= \pi(x)E_{\omega}$ and $\omega(x)=\langle \zeta_n, \pi(x)E_{\omega}\zeta_n \rangle$ for all $x \in \cla$ and $\zeta_n$ is an unit vector in $\clh$.
If $\omega_n \perp \omega$ and $\omega_n \raro \omega$ then $E_{\omega_n} \raro E_{\omega}$ in weak operator topology. So $E_{\omega}E_{\omega'}=0$ for all $\omega' \perp \omega$. This clearly shows that $\Omega^{\perp}$ is a closed set.

\vsp
Thus the set $\{g \in SU_2(\IC): \omega'' \beta_{u(g)} \beta_{r_{\zeta}} = \omega' \beta_{u(g)} \}$ is both open and closed and thus equal to $SU_2(\IC)$ being connected.

\vsp
By (a), we have $\omega = \omega''$ and thus $\omega$ is also an extremal element in $S_{\theta,\IZ_2}$. This completes the proof of (b). The last statement (c) is also obvious now by (a).

\end{proof}

\vsp
In Proposition 3.5, if $\omega$ were an ergodic states of $\IM_R$, then by Proposition 2.6, for an extremal element $\psi \in K_{\omega}$, we could have concluded directly that $\psi \beta_{u(g)} = \psi$ on $\clo_d$ for all $g \in SU_2(\IC)$ by Proposition 2.6. Thus in particular, we could have concluded that $\psi \beta_{\mu I_d} = \psi \beta_{r_{\zeta}}^2 
=\psi$. So $\mu \in H$. The following proposition says much more when $\omega$ is having 
some additional properties.

\vsp 
\begin{pro} 
Let $\omega$ be an extremal element in the convex set $S_{\theta,\IZ_2,+}$ then the following statements are true:

\vsp 
\NI (a) If $\omega$ is also an extremal element in $S_{\theta}$ then $\omega$ is a factor state of $\IM$; 

\vsp 
\NI (b) If $d$ is an odd integer then $\omega$ is a factor state of $\IM$.
\end{pro}
   
\vsp 
\begin{proof} 
We fix a unitary element $w$ in the centre $\clz_{\psi_0}$ of $\pi_{\psi_0}(\tilde{\mbox{UHF}}_d \otimes \mbox{UHF}_d)''$ and an extremal element $\psi \in K_{\omega}$.  
By Proposition 2.1.there exists $z_{w} \in S^1$ so that $\psi Ad_{w} = \psi \beta_{z_{w}}$ 
on $\clo_d$. We compute the following equalities for $x \in \tilde{\clo}_d \otimes \clo_d$:
$$\psi_0 Ad_w \clj_{r_{\zeta}} Ad_{w}(x^*)$$
$$=\psi_0 \beta_{z_{w}} \clj_{r_{\zeta}} Ad_{w}(x^*)$$
$$=\psi_0 \clj_{r_{\zeta}}\beta_{\bar{z}_w} Ad_{w}(x^*))$$
$$=\psi_0 Ad_{w} \beta_{\bar{z}_w}(x))$$
$$=\psi_0 \beta_{z_w} \beta_{\bar{z}_w}(x)$$
$$=\psi_0 (x)$$
$$=\psi_0 \clj_{r_{\zeta}}(x^*)$$
So we have
$$\psi_0((Ad_w \clj_{r_{\zeta}} Ad_{w}(x))) = \psi_0(\clj_{r_{\zeta}}(x))$$
for all $x \in \tilde{\clo}_d \otimes \clo_d$

\vsp 
Since $\psi_0$ is a pure state of $\tilde{\clo}_d \otimes \clo_d$ by Proposition 4.3 in \cite{[23]} ( we need only extremal property of $\omega$ to show $\{x \in \clb(\clk): \tau(x)=\tilde{\tau}(x)=x \}$ is $\IC I$ ), there exists $c \in S^1$ such that 
$$w \clj_{r_{\zeta}} w \zeta_{\psi} = c \clj_{r_{\zeta}}\zeta_{\psi}$$
where anti-unitary operator $\clj_{r_{\zeta}}$ on $\tilde{\clh} \otimes_{\clk} \clh$ is given 
$\clj_{r_{\zeta}}(X)=\clj_{r_{\zeta}}X\clj^*_{r_{\zeta}}$
in Proposition 3.1.

\vsp 
We compute further that 
$$\clj_{r_{\zeta}}w\clj_{r_{\zeta}}\zeta_{\psi} = c w^*\zeta_{\psi}$$
Since both $w^*$ and $\clj_{r_{\zeta}}w\clj_{r_{\zeta}}$ are elememts in the centre
$\clz_{\psi_0}$ we have
$$\clj_{r_{\zeta}}w\clj_{r_{\zeta}}F_0 = c w^*F_0$$
where
$F_0=[\pi_{\psi_0}(\mbox{UHF}_d \otimes \mbox{UHF}_d)''\zeta_{\psi}]$.

\vsp
Using $\clj_{r_{\zeta}}^2=I$ on $\pi_{\psi_0}(\tilde{\mbox{UHF}}_d \otimes \mbox{UHF}_d)''$ for odd or even values of $d$, we also have
$wF_0=\bar{c} \clj_{r_{\zeta}}w^*\clj_{r_{\zeta}}F_0$ and so
$$w^*F_0= (\bar{c} \clj_{r_{\zeta}}w^*\clj_{r_{\zeta}})^*F_0$$
$$= c \clj_{r_{\zeta}} w \clj_{r_{\zeta}}F_0$$
Thus $c=\bar{c}$ i.e. $c=1$ or $-1$. The state $\omega$ being reflection
positive with twist $r_{\zeta}$, we have
$\omega(\clj_{r_{\zeta}}w\clj_{r_{\zeta}}w) \ge 0$ and thus $c=1$.

\vsp
By Stone's spectral theorem, we have
$$\clj_{r_{\zeta}}(E)F_0=EF_0$$
for all projection $E \in \clz_{\psi_0}$ affiliated to $w$.
Unitary element $w$ being an arbitrary element in the centre $\clz_{\psi_0}$ of $\pi_{\psi_0}(\tilde{\mbox{UHF}}_d \otimes \mbox{UHF}_d )''$, we conclude that
$$\clj_{r_{\zeta}}(E)F_0=EF_0$$
for all projection $E$ in $\clz_0$.

\vsp 
As $\Lambda$ takes $\clz_{\psi}$ to itself, we also get 
$$\clj_{r_{\zeta}}(\Lambda(E))F_0=\Lambda(E)F_0$$
for all projection $E \in \clz_{\psi_0}$. Since $\clj_{r_{\zeta}}(\Lambda(X))=\tilde{\Lambda}(\clj_{r_{\zeta}}(X))$, we arrive at
$\tilde{\Lambda}(E)F_0=\Lambda(E)F_0$.

\vsp 
The action $X \raro S^*_iXS_i$ also takes elements in $\clz_{\psi}$ to itself and so we may take element 
$S^*_iES_i$ in place of $E$ for 
$$\;\tilde{\Lambda}(S^*_iES_i)F_0$$
$$= \Lambda(S_i^*ES_i)F_0 $$
$$=\sum_{1 \le k  \le d}S_kS_i^*ES_iS_k^*F_0$$
$$=\sum_{1 \le k \le d}S_kS_i^*S_iS_k^*EF_0$$
$$=EF_0$$

\vsp 
However the map $XF_0 \raro \tilde{\Lambda}(S_i^*XS_i)F_0$ is the left shift on
$\tilde{\mbox{UHF}}_d \otimes \mbox{UHF}_d$ and so 
we conlcude that 
$$\theta^{-1}(E)F_0=EF_0$$
i.e. $\theta(E)F_0=EF_0$.

\vsp 
Since $\theta(F_0)=F_0$, where we recall $F_0=[\pi_{\psi_0}(\mbox{UHF}_d \otimes \mbox{UHF}_d)''\zeta_{\psi}]$, we have
$$\theta(E_0)=E_0,$$
where $E_0=EF_0$ is a projection in the centre of $\pi_{\omega}(\IM)''$. The state $\omega$ being an ergodic state of
$\IM$, we conclude $E_0$ is either $0$ or $F_0$. Since there is an isomorphism between the centre of $\pi_{\omega}(\IM)''$ with $\clz_{\psi_0}F_0$, we conclude that $\omega$ is a factor state of $\IM$. This completes the proof of (a).

\vsp 
We will modify our proof for (a) in order to prove (b). We will prove that $(\clo_d),\Lambda,\psi)$ is ergodic and thus its restriction to $(\mbox{UHF}_d,\Lambda,\omega_R)$ is also ergodic. This will prove (b) by (a). 

\vsp 
To that end, let $E'$ be a projection in the centre of $\pi_{\psi}(\clo_d)''$ equivalently 
$\Lambda(E')=E'$. So the projection $E=\clj_{r_{\zeta}}(E')E'$ is in the centre of $\pi_{\psi_0}(\tilde{\clo}_d \otimes \clo_d)''$ and $\clj_{r_{\zeta}}(E)=E$ for odd values of $d$.

\vsp 
We recall that $\theta(X)=S^*XS$ on $X \in \pi_{\psi_0}(\tilde{\mbox{UHF}}_d \otimes \mbox{UHF}_d)''$, where $S=\sum_{1 \le i \le d} \tilde{S}_iS_i^*$ is a unitry operator by Cuntz relations (16) with its inverse $S^*=\sum_{1 \le i \le d}S_i\tilde{S}^*$. We 
verify that 
$$\theta(E')=\Lambda(E')=E'$$ 
and 
$$\theta(\clj_{r_{\zeta}}(E'))$$
$$=\clj_{r_{\zeta}}(\theta^{-1}(E'))$$
$$=E'$$
So $\theta(E)=E$ and the linear map $$\psi_E: X \raro \psi(X E)$$
on $\pi_{\psi_0}(\tilde{\mbox{UHF}}_d \otimes \mbox{UHF}_d)''$ satisfies 
$$\psi_E(\clj_{r_{\zeta}}(X^*))$$
$$=\psi(\clj_{r_{\zeta}}(X^*)E)$$
$$=\psi(\clj_{r_{\zeta}}(X^*)\clj_{r_{\zeta}}(E))$$
$$=\psi(\clj_{r_{\zeta}}(X^*E))$$
$$=\psi(XE)$$
$$=\psi_E(X)$$
Since $\omega=\psi_{|}\mbox{UHF}_d$ is an extremal element in $S_{\theta,\IZ_2}$, we have
$$\psi_E(X) = \psi(X) \psi(E)$$ 
for all $X \in \pi_{\psi}(\tilde{\mbox{UHF}}_d \otimes \mbox{UHF}_d))''$ and $\psi(E)$ is equal to either $0$ or $1$. 

\vsp 
The state $\omega$ being reflection positive with twist $r_0$ on $\IM$ and $SU_2(\IC)$, 
we have 
$$\omega(\clj_{r_{\zeta}}(X)X) \ge 0$$ for all $X \in \pi_{\psi}(\mbox{UHF}_d)''$.  

\vsp
We consider now the map
$(X,Y) \raro \omega(\clj_{r_{\zeta}}(X)Y)$ that is linear in right side and conjugate linear in the left side variables on $\pi_{\psi}(\mbox{UHF}_d)''$.
The map $$(X,Y) \raro \psi(\clj_{r_{\zeta}}(E'X)E'Y)=\psi(E)\psi(\clj_{r_{\zeta}}(X)Y)$$
is a pre-inner product 
$$\langle \langle X \zeta_{\psi},Y \zeta_{\psi} \rangle \rangle_{E'}=\omega(\clj_{r_{\zeta}}(E'X)E'Y)$$ 
on the vector space $[\pi_{\psi}(\mbox{UHF}_d)''\zeta_{\psi}].$  

\vsp 
However we computed above that $\omega(\clj_{r_{\zeta}}(X)\clj_{r_{\zeta}}(E')E'X) = 0$ by the first part if $\psi(\clj_{r_{\zeta}}(E')E')=0$. In such case, we have 
$E'X\zeta_{\psi}=0$ for all $X \in  \pi_{\psi}(\mbox{UHF}_d)''$. Similarly, if $\psi(E'\clj_{r_{\zeta}}(E'))=1$, then $\psi(E')=1=\psi(\clj_{r_{\zeta}}(E'))$ and so by replacing the role of $E'$ by $I-E'$, we arrive at $(I-E')X\psi_{\zeta}=0$ for all $X \in \pi_{\psi}(\mbox{UHF}_d)''$. Thus $E'$ is either $0$ or $I$ on $\pi_{\psi}(\mbox{UHF}_d)''\zeta_{\psi}]$. So $\omega$ is an extremal element in $S_{\theta}$ and by (a),
$\omega$ is a factor state. 

\end{proof}

\vsp 
In the following theorem, we sum up our main result required as our appliation in section 5 and 6. Let $S_{\theta,G,+}$ be convex subset of $S_{\theta,G}$ consist of 
reflection positive states with twist $\beta_{r_{\zeta}}$, where $G=SU_2(\IC) \otimes \IZ_2$. Similarly, we also set $S_{\theta,\IZ_2,+}$ for subset of $S_{\theta,\IZ_2}$ consist of reflection positive states with twist $\beta_{r_{\zeta}}$.

\vsp 
\begin{thm} The following statements hold:

\vsp 
\NI (a) Let $\omega \in S_{\theta,\IZ_2,+}$ and $d$ be an odd integer then there exists the following direct integral factor decomposition: 
$$\omega = \int^{\oplus} \omega_{\alpha}d\mu_{\omega}(\alpha)$$
for some Borel probability measure $\mu_{\omega}$, where factors 
$\omega_{\alpha}$ are elements in $S_{\theta,\IZ_2,+}$ for $\mu_{\omega}$-almost 
everywhere;
 
\vsp 
\NI (b) Let $\omega$ be an element in $S_{\theta,G,+}$ and 
$$\omega = \int \omega_{\alpha}d\mu_{\omega}(\alpha)$$
be an extremal decomposition in the convex set $S_{\theta,G,+}$, 
where $\mu_{\omega}$ is a Borel probability measure $\mu_{\omega}$. Then the following statements hold true:

\vsp 
\NI (i) For $\mu$-almost everywhere $\omega_{\alpha}$ are also extremal elements in the convex set $S_{\theta,\IZ_2}$;

\vsp 
\NI (ii) If $d$ is an odd integer then extremal decomposition given in (b) is a factor decomposition.  

\vsp 
\NI (c) Let $\omega$ be an element in $S_{\theta,G,_+} \bigcap \{\omega: \omega = \tilde{\omega} \}$ and 
$$\omega = \int \omega_{\alpha}d\mu_{\omega}(\alpha)$$
be an extremal decomposition in the convex set $S_{\theta,G,+} \bigcap \{\omega:\tilde{\omega}=\omega \}$, where $\mu_{\omega}$ is a Borel probability measure. 
Then the following statements hold true:
\vsp 
\NI (i) For $\mu$-almost everywhere $\omega_{\alpha}$ are also extremal elements in the convex set $S_{\theta,\IZ_2} \bigcap \{\omega: \tilde{\omega} = \omega \}$; 

\vsp 
\NI (ii) If $d$ is an odd integer then extremal decomposition given in (b) is a factor decomposition.  
\end{thm} 

\vsp 
In general the statement (ii)  in Theorem 3.7 (b) is false for even values of $d$ and counter examples are included in section 5. Alternatively, in Theorem 3.7 (a) 
$\omega_{\alpha}$ need not be $SU_2(\IC)$ invariant even when $\omega$ is an element in $S_{\theta,G,+}$ unless $d$ is an odd integer.

\vsp 
\begin{proof} 
For (a), we consider extremal decomposition of $\omega$ in $S_{\theta,\IZ_2,+}$ and apply Proposition 3.6 (b). For (b), we consider extremal decomposition of $\omega$ in $S_{\theta,G,+}$ and apply Proposition 3.5 (d) and then apply Proposition 3.6 (b) for the required result.  For (c) we verify that $\tilde{\omega \beta_{u(g)}}=\tilde{\omega} \beta_{u(g)}$ for all $g \in SU_2(\IC)$ for any $\omega \in S_{\theta}$. Result follows once we
restrict our method employed to prove (a) for the convex subset $S_{\theta,G,+} \bigcap \{ \omega: \omega = \tilde{\omega} \}$. We omit the details.
\end{proof} 

\vsp 
\begin{cor} 
Let $\omega \in S_{\theta,G,+} \bigcap \{\omega:\tilde{\omega} = \omega \}$. Then $\omega$ is a real i.e. $\omega=\bar{\omega}$, where $\bar{\omega}$ is defined in (9).
If $\omega=\int \omega_{\alpha}d\mu(\alpha)$ is an extremal decomposition in the convex set 
$S_{\theta,G,+} \bigcap \{\omega:\tilde{\omega} = \omega \}$ then $\omega_{\alpha}$ in the direct integral factor decomposition of $\omega$ given in Theorem 3.7 (c) are real states 
for $\mu_{\omega}$ almost everyhwhere.    
\end{cor} 

\vsp 
\begin{proof}  
Since $r_{\zeta} \in u(h)$ for some $h \in SU_2(\IC)$ and $\omega$ is $SU_2(\IC)$ invariant, we have $\omega = \omega \beta_{r_{\zeta}}$. As 
$\omega =\tilde{\omega}$ and $\bar{\omega} = \tilde{\omega} \beta_{r_{\zeta}}$ ($\omega(x)=\omega(\clj_{r_{\zeta}}(x^*)$), we have 
$\bar{\omega}=\omega$.  
\end{proof} 

\vsp 
For an integer $m \ge 1$, we consider the two-sided lattice $\IZ_m \otimes \IZ$ with $m$-many legs, where $\IZ_m=\{0,1,..,m-1\}$ and the group action $\psi \raro \psi \beta_{u^m(g)}$, where $u^{(m)}(g)=u(g) \otimes u(g) ..\otimes u(g)$ -$m$-fold tensor product. So we have 
$$\bar{u^{(m)}(g)} = r_0^{(m)}u^{(m)}(g)r_0^{(m)}$$
and 
$$(\clj^{(m)}_{r_{\zeta}})^2 = \beta_{\mu^m}$$
  
\vsp 
\begin{thm} 
We consider the lattice $\IZ_m \otimes \IZ$ and associate UHF$_d$ algebra $\IM_{\IZ_m \otimes \IZ}$. Let $\omega$ be an element in $S_{\theta,G,+}$ and 
$$\omega = \int \omega_{\alpha}d\mu_{\omega}(\alpha)$$
be an extremal decomposition in the convex set $S_{\theta,G,+}$, where $\mu_{\omega}$ is a Borel probability measure $\mu_{\omega}$. Then 

\vsp 
\NI (i) For $\mu$-almost everywhere $\omega_{\alpha}$ are also extremal elements in the convex set $S_{\theta,\IZ_2}$;

\vsp 
\NI (ii) If $\mu^m=1$ then the extremal decomposition is a factor decomposition.  

\end{thm} 

\vsp
\begin{proof}
An easy adaptation of Proposition 3.6 gives a proof.
\end{proof}

\section{Integer spin $s$ and half-odd integer spin $s$ ($2s+1=d)$:}

\vsp 
Let $\omega$ be an extremal element in the convex set of translation invariant states of $\IM$ and $\omega$ is real and lattice symmetric with twist
$\beta_{r_0}$ i.e. $\tilde{\omega} \beta_{r_0}=\bar{\omega}=\omega$. So
$\omega \in S_{\theta,\IZ_2}$ since $\bar{\omega}=\tilde{\omega}\beta_{r_{\zeta}}$ on $\IM$.

\vsp
As in Theorem 3.5 in \cite{[22]}, we fix an extremal element $\psi \in K_{\omega}$.
Since $\tilde{\omega}\beta_{r_0}=\omega$, $\tilde{\psi} \beta_{r_0} \in K_{\omega}$
and the element is also an extremal in $K_{\omega}$, So by Proposition 2.2, $\tilde{\psi} \beta_{r_0}=\psi \beta_{\zeta_0}$ for some $\zeta_0 \in S^1$. We use
$\tilde{\tilde{\psi}}=\psi$ and $r_0^2=I_d$ to conclude $\zeta_0^2 \in H$. Besides,
$\tilde{(\psi \circ \beta_z)}=\tilde{\psi} \circ \beta_z$ for all $z \in S^1$ and so $\tilde{\psi} \beta_{r_0} = \psi \beta_{\zeta_0}$ holds for all $\psi \in K_{\omega}$.

\vsp
Unlike reflection symmetry with twist $r_0$, we have $\bar{\psi \beta_z}=\bar{\psi} \beta_{\bar{z}}$ for any element $\psi \in K_{\omega}$. We aim to choose an extremal element $\psi \in K_{\omega}$ so that $\bar{\psi}=\psi \zeta_0$. We fix an arbitary extremal element $\psi'$ in $K_{\omega}$ and get $\bar{\psi}=\psi \beta_{z_0}$ for some $z_0 \in S^1$ so that $z_0^2 \in H$. If so then we check that
$$\bar{\psi'\beta_z}=\bar{\psi'} \beta_{\bar{z}} = \psi \beta_{z_0 \bar{z}}$$
for all $z \in S^1$. We choose $\psi = \psi' \beta_z$ for which $\bar{z}^2z_0=\zeta_0$. Such a choice for $z \in S^1$ is possible and so we get the required relation $\bar{\psi}=\psi \beta_{\zeta_0}$.

\vsp
Thus as in Theorem 3.5 in \cite{[22]}, there exists an extremal element $\psi \in K_{\omega}$ such that $\tilde{\psi} \beta_{r_0}=\bar{\psi}=\psi \beta_{\zeta_0}$ for some $\zeta_0 \in S^1$ such that $\zeta_0^2 \in H$. Furthermore, if $\omega$
is also pure then Popescu elememts $(\clk,v_i,1\le i \le d, \zeta_{\psi})$ of
$\psi$ in the support projection $\clk=[\pi_{\psi_0}(\clo_d)'\zeta_{\psi_0}]$ described as in Proposition 2.3 then their exists a unitary operator $v_{r_0}$
on $\clk$ satisfying the following properties:

\vsp
\NI (a) $v_{r_0}\zeta_{\psi}=\zeta_{\psi}$ and $v_{r_0}$ commutes with $\clj$
and $\Delta^{1 \over 2}$;

\vsp
\NI (b) $v_{r_0}\beta_{r_0}(\tilde{v}_{I'}\tilde{v}_{J'}^*v_Iv_J^*)v_{r_{\zeta}}^*=\clj v_{I'}v_{J'}^*\tilde{v}_I\tilde{v}_J \clj$ for all $|I'|,|J'|I|$ and $|J|
< \infty$;

\vsp
\NI (c) $v_{r_0}u_z=u_{\bar{z}}v_{r_0}$;

\vsp
\NI (d) $v_{r_0}$ is self-adjoint if and only if $\bar{r_0}=r_0$.

\vsp
So such an element $\psi$ is also an element in $K_{\omega, \IZ_2}$ i.e. $\bar{\psi} = \tilde{\psi} \beta_{r_{\zeta}}$ if $\zeta \in H$ and the elememt
is obviously an extremal point in $K_{\omega,\IZ_2}$.

\vsp
Let $\omega$ be also $SU_2(\IC)$ invariant. By Proposition 2.6, we have $\psi_0 \beta_{u(g)}=\psi_0$ for all $g \in SU_2(\IC)$ on $\tilde{\clo}_d \otimes \clo_d$. Since $r_{\zeta}=u(i\sigma_y)$, there exists a unitary operator $\hat{r}_{\zeta}:\tilde{\clh} \otimes_{\clk} \clh \raro \tilde{\clh} \otimes_{\clk} \clh$ such that $\hat{r}_{\zeta} \zeta_{\psi}=\zeta_{\psi}$ and
$$Ad_{\hat{r}_\zeta}(\pi_{\psi_0}(x))= \pi_{\psi_0}(\beta_{r_{\zeta}}(x))$$
for all $x \in \tilde{\clo}_d \otimes \clo_d$.

\vsp 
\begin{pro} 
Let $\omega,\psi$ be as in Proposition 3.1 and $\omega$ be also $SU_2(\IC)$ invariant. We consider the anti-automorphism $\clj_{r_{\zeta}}$ on $\tilde{\clo}_d \otimes \clo_d$ and its induced anti-automorphism map
$\clj_{r_{\zeta}}$ on $\pi_{\psi_0}(\tilde{\clo}_d \otimes \clo_d)''$
defined by
\be
\clj_{r_{\zeta}}(\pi_{\psi_0}(x)) = \pi_{\psi_0}(\clj_{r_{\zeta}}(x))
\ee
for all $x \in \tilde{\clo}_d \otimes \clo_d$ as in Proposition 3.1.

If $\omega$ is also pure then $\hat{\clj}_{r_{\zeta}}(P)=P$ and the corner anti-automorphism, defined by
$$\hat{\clj}_{r_{\zeta}}(a) = P\hat{\clj}_{r_{\zeta}}(PaP)P$$
for all $a \in \clb(\clk)$ satisfies the following: 
\be 
\hat{\clj}_{r_{\zeta}}(a) = \clj \gamma \hat{r}_{\zeta} a  
\hat{r}_\zeta^* \gamma^* \clj 
\ee 
for all $a \in \clb(\clk)$. Furthermore, we have the following consequences: 

\vsp 
\NI (a1) $\hat{\clj}^2_{r_{\zeta}}=\beta_{\mu}$;

\vsp 
\NI (a2) $\beta_{\bar{r_{\zeta}}}(\tilde{S}_I\tilde{S}_J^*S_{I'}S_{J'}^*) \clj_{\gamma} \hat{r_{\zeta}} = \clj_{\gamma}\hat{r_{\zeta}} S_IS_J^*\tilde{S}_{I'}\tilde{S}^*_{J'}$ for all $|I'|,|J'|,|I|$ and $J| < \infty$.

\vsp 
\NI (a3) $Ad_{\hat{U}(g)} \hat{\clj}_{r_{\zeta}}=\hat{\clj}_{r_{\zeta}} Ad_{\hat{U}(g)}$for all $g \in SU_2(\IC)$;

\vsp 
\NI (b1) $Ad^2_{\gamma_{r_{\zeta}}}=\beta_{\zeta^2 I_d}$, where 
$\gamma_{r_{\zeta}}= \gamma \hat{r}_\zeta$ commutes with modular elements $\clj$ and $\Delta^{1 \over 2}$;

\vsp 
\NI (b2) $\beta_{\bar{r_{\zeta}}}(\tilde{v}_I\tilde{v}_J^*v_{I'}v_{J'}^*)  \clj \gamma_{r_{\zeta}} = \clj \gamma_{r_{\zeta}} v_Iv_J^*\tilde{v}_{I'}\tilde{v}^*_{J'}$ for all $|I'|,|J'|,|I|$ and $J| < \infty$;

\vsp 
\NI (b3) $\gamma_{r_{\zeta}}$ commutes the representation $\{\hat{u}(g):g \in SU_2(\IC) \}$;

\vsp 
There exists a unique unitary operator $\Gamma_{r_{\zeta}}$ and an anti-unitary operator extending $\clj$ on $\tilde{\clh} \otimes_{\clk} \clh$ extending unitary $\gamma_{r_{\zeta}}:\clk \raro \clk$ and anti-unitary operator $\clj:\clk \raro \clk$ respectively such that

\vsp 
\NI (c1) $Ad_{\Gamma_{r_{\zeta}}}^2=\beta_{\zeta^2}$; $Ad_{\Gamma_{r_{\zeta}}}$ acts on $\pi(\clo_d)''$ and $\pi(\mbox{UHF}_d)''$ ( $\pi(\tilde{\clo}_d)''$ and 
$(\pi(\tilde{\mbox{UHF}}_d)''$ ) respectively; 

\vsp 
\NI (c2) $\beta_{\bar{r_{\zeta}}}(\tilde{S}_I\tilde{S}_J^*S_{I'}S_{J'}^*)  \clj \Gamma_{r_{\zeta}} = \clj \Gamma_{r_{\zeta}} S_IS_J^*\tilde{S}_{I'}\tilde{S}^*_{J'}$ 
for all $|I'|,|J'|,|I|$ and $J| < \infty$;

\vsp 
\NI (c3) $\Gamma_{r_{\zeta}}$ and $\clj$ commutes the representation $\{\hat{U}(g):g \in SU_2(\IC) \}$;

\end{pro}

\vsp 
\begin{pro} 
Let $\omega$ and $\psi$ be as in Proposition 3.1 and $\omega$ is also pure and 
reflection positive with twist $r_0$. Then we have 
$Ad_{\gamma_{r_{\zeta}}}(a)=a$ for all $a \in \clm_0$, where $\clm_0=\{ a \in \clm: \beta_z(a)=a \forall z \in H \}=P\pi_{\psi}(\mbox{UHF}_d)''P$.
\end{pro}

\vsp 
\begin{proof} 
We recall that $r_{\zeta}=\zeta r_0$ and thus Theorem 3.5 (d) in \cite{[23]} gives a proof as $$Ad_{\gamma_{r_{\zeta}}}=Ad_{\gamma_{r_0}}$$ 
on $\clm_0$, where $Ad_{\gamma_0}$ is defined in Theorem 3.5 in \item{[23]} using invariance property 
$$\psi_0 \beta_{r_0} \otimes \beta_{r_0}=\psi_0$$ 
on $\IM$, which is identified as before with $\tilde{\mbox{UHF}_d} \otimes\mbox{UHF}_d$. Alternatively, we can use the same argument directly with $Ad_{r_{\zeta}}$ using our hypothesis that $\omega$ is reflection positive with the twist $r_0$ to conclude $Ad_{r_{\zeta}}$ on $\pi(\mbox{UHF}_d)''$ is $I$.  
\end{proof}

\vsp 
\begin{pro} 
Let $\omega$ be a translation invariant real lattice symmetric pure state of $\IM$ with a twist $r_0$ and $\gamma_{r_{\zeta}}$ be the unitary operator described in Proposition 4.1. If $\omega$ is also reflection positive with the twist $r_0$ then
\be 
\langle \Delta^{-{1 \over 2}}v_i^*\zeta_{\psi},  \beta_{r_0}(v_j^*) \zeta_{\psi} \rangle
=\langle \Delta^{1 \over 2} v_j^*\zeta_{\psi}, \beta_{r_0}(v_i^*) \zeta_{\psi} \rangle
\ee 
\be 
\langle \Delta^{-{1 \over 2}} \sigma_t(v_i^*) \zeta_{\psi}, \beta_{r_0}(\sigma_s(v_j^*)) \zeta_{\psi} \rangle
=\langle \sigma_s(v^*_j) \zeta_{\psi}, \Delta^{1 \over 2} \beta_{r_0}(\sigma_t(v_i^*)) \zeta_{\psi} \rangle \forall s,t \in \IR
\ee
\be 
\langle \Delta^{-{1 \over 2}}  (v^{\delta_1}_i)^*  \zeta_{\psi}, \beta_{r_0}((v^{\delta_2}_j)^*) \zeta_{\psi} \rangle
=\langle (v^{\delta_2}_j)^*\zeta_{\psi}, \Delta^{1 \over 2} \beta_{r_0}((v^{\delta_1}_i)^*) \zeta_{\psi} \rangle \forall \delta_1,\delta_2 > 0
\ee
and 
\be 
\langle \zeta_{\psi}, v^{\delta_1}_i \sigma_{i(y-{1 \over 2})}(\beta_{r_0}((v^{\delta_2}_j)^*)) \zeta_{\psi} \rangle
=\langle \zeta_{\psi}, v^{\delta_2}_j \sigma_{i(y+{1 \over 2})}(\beta_{r_0}((v^{\delta_1}_i)^*))\zeta_{\psi} \rangle
\ee
for all $1 \le i,j \le d,\;t \in \IR$ and $s=1$, where 
$$x^{\delta} = { 1 \over \sqrt{2\pi}} \int_{\IR} exp^{ -{1 \over 2} {t^2 \over \delta^2}} \sigma_t(x)dt$$ for $\delta> 0$ and $z \raro \sigma_z(x)$ is the analytic extension of $t \raro \sigma_t(x)$ for an analytic element $x \in \clm$ to $\IC$.  
\end{pro}

\vsp 
\begin{proof} 
We recall, $x\zeta_{\psi}$ is an element in the domain of $\Delta^{1 \over 2}$ for $x \in \clm$ and $y\zeta_{\psi}$ is an element in the domain of $\Delta^{-{1 \over 2}}$ for $y \in \clm'$ [6]. We also recall that $v_i^*\zeta_{\psi} = \tilde{v}_i^*\zeta_{\psi}$ and thus both sides of equalities in (46)-(49) are well defined. We need to establish those equalities.

\vsp 
The element $\sigma_t(v_i) \clj \beta_{\bar{r_{\zeta}}}(\sigma_s(\tilde{v}_j)) \clj \in \clm_0$, i.e. $(\beta_z:z \in H )$ invariant element in $\clm$ as $\beta_z(v_i)=zv_i$ and $\beta_z(\clj \beta_{\bar{r_{\zeta}}}(\tilde{v}_j) \clj) = \bar{z} \clj \beta_{\bar{r_{\zeta}}}(\tilde{v}_j) \clj$, where we used commuting property of $(\beta_z;z \in H )$ with the modular group $(\sigma_t)$ on $\clm$ and commuting property of $(u_z:z \in H)$ with $\clj$ as $\beta_z=Ad_{u_z}$ and $\beta_z(\tilde{v}_j) = z\tilde{v}_j$ for $z \in H$. 

By our hypothesis that $\omega$ is reflection positive with twist $r_0$, any element in $\clm_0$ is $Ad_{\gamma_{r_{\zeta}}}$ invariant by Proposition 4.2. So we have the following equality for any $1 \le i,j \le d$:
$$\sigma_t(v_i) \clj \beta_{\bar{r_{\zeta}}}(\sigma_s(\tilde{v}_j))\clj$$
$$=Ad_{\gamma_{\bar{r_{\zeta}}}}(\sigma_t(v_i))Ad_{\gamma_{r_{\zeta}}}(\clj\beta_{\bar{r_{\zeta}}}(\sigma_s(\tilde{v}_j))\clj)$$
$$=\clj \beta_{\bar{r_{\zeta}}}(\sigma_t(\tilde{v}_i))\clj \sigma_s(v_j)$$ 
where we used again modular group commutes with any automorphism that preserves the faithful normal state $\phi$ on $\clm$. 
Since $r_{\zeta}=\zeta r_0$, we get
\be 
\sigma_t(v_i) \clj \beta_{\bar{r_0}}(\sigma_s(\tilde{v}_j)) \clj = \clj \beta_{\bar{r_0}}(\sigma_t(\tilde{v}_i))\clj \sigma_s(v_j)
\ee
for all $s,t \in \IR$. 
\vsp 
Now we use $x^*\zeta_{\psi}=Sx\zeta_{\psi} =\clj \Delta^{1 \over 2}x\zeta_{\psi}$ for
$x \in \clm$ and $y^*\zeta_{\psi} = \clj \Delta^{-{1 \over 2}}y\zeta_{\psi}$ for $y \in \clm'$ to verify (46) as we can compute the following equalities at ease:
$$\langle \Delta^{-{1 \over 2}} \tilde{v}_i^* \zeta_{\psi}, \beta_{r_0}(v_j^*) \zeta_{\psi}\rangle $$
$$=\langle  v_i^* \zeta_{\psi},\Delta^{-{1 \over 2}} \beta_{r_0}(v_j^*) \zeta_{\psi}\rangle$$
$$=\langle  v_i^* \zeta_{\psi},\Delta^{-{1 \over 2}} \beta_{r_0}(\tilde{v}_j^*) \zeta_{\psi}\rangle $$
$$=\langle  v_i^* \zeta_{\psi},\clj \clj \Delta^{-{1 \over 2}} \beta_{r_0}(\tilde{v}_j^*) \zeta_{\psi}\rangle $$
$$=\langle \zeta_{\psi}, v_i \clj \beta_{\bar{r_0}}(\tilde{v}_j) \clj \zeta_{\psi} \rangle $$
( now by (50) )
$$=\langle \zeta_{\psi}, \clj \beta_{\bar{r_0}}(\tilde{v}_i) \clj v_j \zeta_{\psi} \rangle $$
$$=\langle \clj \beta_{\bar{r_0}}(\tilde{v}^*_i) \clj \zeta_{\psi}, v_j\zeta_{\psi} \rangle$$
$$=\langle \clj \beta_{\bar{r_0}}(v^*_i) \zeta_{\psi}, v_j\zeta_{\psi} \rangle$$
$$=\langle \clj v_j  \zeta_{\psi}, \beta_{r_0}(v^*_i)\zeta_{\psi} \rangle$$
(we used conjugate linear property of $\clj$)
$$=\langle \clj \clj \Delta^{1 \over 2}v^*_j  \zeta_{\psi}, \beta_{r_0}(v^*_i)\zeta_{\psi} \rangle$$
$$=\langle \Delta^{1 \over 2} v_j^* \zeta_{\psi}, \beta_{r_0}(v^*_i) \Delta^{-{1 \over 2}}\zeta_{\psi} \rangle$$
We can verify (47) along the same line. The equality (48) is a simple consequence of (47).
For (49), we recall that $x^{\delta}$ is an analytic element for the modular group $(\sigma_t)$ for any $x \in \clm$ or $\clm'$ and $\delta > 0$.
\end{proof}

\vsp 
\begin{thm}
Let $g \raro v^i_j(g)$ be an irreducible representation of $SU_2(\IC)$ and the state $\omega$ in Proposition 4.1 be also $SU_2(\IC)$ invariant. Then the following holds:

\NI (a) $d$ is an odd integer;

\NI (b) $\Delta=I$ and $\clm=\clm_0$ is a finite type-I factor and $\phi$ is the normalised trace on $\clm$; 

\NI (c) $H$ is the trivial subgroup of $S^1$ and $Ad_{\gamma_{r_{\zeta}}}=I$; 

\NI (d) $v_i^* = \beta_{r_{\zeta}}(v_i)$;

\NI (e) There exists an irreducible representation $g \raro \hat{u}(g) \in \clm$ such that 

\be 
\hat{u}(g)v_i^*\hat{u}(g)^*=\beta_{u(g)}(v_i^*)
\ee
and representation $g \raro \hat{u}(g)$ is an odd integer or even integer representation of $SU_2(\IC)$. 
\end{thm}

\begin{proof}

By Clebsch-Gordan theory valid for irreducible representation $g \raro u(g)$
of the group $SU_2(\IC)$, the representation $g \raro u(g) \otimes \bar{u(g)}$ 
in $\IC^d \otimes \IC^d$ admits a unique invaiant subspace. The state $\omega$ being $SU_2(\IC)$-invariant, the vectors $\langle \zeta_{\psi}, S_i^*S_j \zeta_{\psi}\rangle$ and $\langle \zeta_{\psi}, S_jS^*_i \zeta_{\psi}\rangle$ are  
$g \raro \bar{u(g)} \otimes u(g)$ invariant by Proposition 2.5 (a) and thus 
$$\langle \zeta_{\psi}, S_iS_j^*\zeta_{\psi} \rangle = {\delta^i_j \over d}$$

\vsp 
The state $\omega$ being $SU_2(\IC)$-invariant, the vector $\langle \zeta_{\psi}, v_i\Delta^{s}v^*_j \zeta_{\psi} \rangle$ is also $g \raro u(g) \otimes \bar{u(g)}$ invariant for any real $s$ since $\Delta$ commutes with $\hat{u(g)}$ for all
$g \in SU_2(\IC)$ by Proposition 2.5 (c). 

\vsp 
For the time being, we fix $\delta > 0$ and simplify notation $v_i^{\delta}$ for $v_i$ and compute that  
\be 
\langle \zeta_{\psi}, v_i\Delta^{y}v_j^* \zeta_{\psi} \rangle = \delta^i_j c_y
\ee
for some positive constant $c_y$ independent of $i,j$. 

\vsp 
Now we use (52) in the equality (49) for $y={1 \over 2}$ to conclude that
$$ 
(r_0)^j_i \langle \zeta_{\psi} v_i v_i^* \zeta_{\psi} \rangle
=(r_0)^i_j \langle \zeta_{\psi} v_j \Delta v_j^* \zeta_{\psi} \rangle
$$ 
for all $1 \le i,j \le d$. Since $r_0^2=I$, i.e. $r_0=r_0^*$ and each row or column vector is non zero, we conclude that
\be 
\langle \zeta_{\psi} v_i v_i^* \zeta_{\psi} \rangle
=\langle \zeta_{\psi}v_j \Delta v_j^* \zeta_{\psi} \rangle
\ee 
for some $i,j$ and hence for all $1 \le i,j \le d$ by (52). Similarly, we also use (52) in the equality (49) for $y=-{1 \over 2}$ to conclude that
\be 
\langle \zeta_{\psi} v_i \Delta^{-1}v_i^* \zeta_{\psi} \rangle
=\langle \zeta_{\psi}v_j v_j^* \zeta_{\psi} \rangle
\ee 
for some $i,j$ and hence for all $1 \le i,j \le d$ by (52).

\vsp 
So we have by (53) and (54)
$$||[\Delta^{1 \over 2}v_i^*-\Delta^{-{1 \over 2}}v_i^*]\zeta_{\psi}||^2$$
$$=||\Delta^{1 \over 2} v_i^*\zeta_{\psi}||^2+||\Delta^{-{1 \over 2}}v_i^*\zeta_{\psi}||^2-2||v_i^*\zeta_{\psi}||^2$$
$$=0$$

\vsp 
By separating property for $\zeta_{\psi}$ for $\clm$, we get
$$\Delta v_i^*\Delta^{-1} = v_i^*$$
i.e. $\Delta$ commutes with each $v_i^*$. 
$\Delta$ being self-adjoint, $\Delta$ also commutes with each $v_i$ i.e. $\Delta$ commutes with each $v_i^{\delta}$ for any $\delta > 0$ once we remove simplified notation. Thus $\Delta$ commutes with each $v_i$ and so $\Delta \in \clm'$. Since $\clj \Delta \clj = \Delta^{-1}$, we also conclude $\Delta \in \clm$. The von Neumann algbera $\clm$ being a factor and $\Delta \zeta_{\psi}=\zeta_{\psi}$,
we conclude $\Delta=I$.

\vsp 
We claim that $\clm$ is a finite type-I factor rather than a type-$II_1$ finite factor.
Suppose not. Then $\clm_0$ is also a type-$II_1$ finite factor. The von-Neumann factor $\clm_0$ being the corner of $\pi_{\omega}(\mbox{UHF}_{d})''$ by $P \in \pi_{\psi}(\mbox{UHF}_d)''$, $\pi_{\omega}(\mbox{UHF}_{d})''$ is also a type-II von-Neumann factor. 

\vsp 
We consider the GNS space $(\clh_{\omega_R},\pi_{\omega_R},\zeta_{\omega_R})$ associated with 
$(\IM_R,\omega_R)$. So $\pi_{\omega_R}(\IM_R)''$ is a type-II factor and $\zeta_{\omega_R}$ is cyclic for $\pi_{\omega_R}(\IM_R)''$ in $\clh_{\omega_R}$. We will rule out the following two possible cases: As in (17), we identify $\IM_R$ with $\mbox{UHF}_d$ with respect to a orthonormal basis $(e_i)$ for $\IC^d$. 

\vsp 
\NI (i) $\pi_{\omega_R}(\mbox{UHF}_d)''$ is a type-II$_1$ factor.  

\vsp 
In such a case $\pi_{\omega_R}(\mbox{UHF}_d)''$ admits a unique tracial state say $\omega_0$ \cite{[Dix]}. Since $\omega_0 \Lambda$ is also a tracial state on $\pi_{\omega_R}(\mbox{UHF}_d)''$, we get by unisqueness of tracial state, $\omega_0 = \omega_0 \Lambda$. But $\omega \Lambda = \omega$ on $\pi_{\omega_R}(\mbox{UHF}_d)''$ and $\omega$ is a factor state, in particular, an ergodic state i.e. unique invariant state of right translation dynamics $(\mbox{UHF}_d,\theta)$. Thus $\omega = \omega_0$ on $\mbox{UHF}_d$. So $\omega$ is the unique trace on $\IM$, contradicting our hypothesis that $\omega$ is pure.  

\vsp 
\NI (ii) $\pi_{\omega_R}(\mbox{UHF}_d)''$ is a type-II$_{\infty}$ factor. 

\vsp 
In this case, $P$ is a finite projection in $\pi_{\omega_R}(\mbox{UHF}_d)''$ and $\clm_0=P\pi_{\omega_R}(\mbox{UHF}_d)''P$ is type-II$_1$ factor and $\pi_{\omega}(\mbox{UHF}_d)''$ is isomorphic to $\clm_0 \otimes \clb(\clh)$, where $\clm_0$ is type-II$_1$ factor acting on $\clk$ and $\clh$ is an infinite dimensional Hilbert space \cite{[Dix]}.  
Note that $\clm_0'$ is also a type-II$_1$ factor and $\pi_{\omega_R}(\mbox{UHF}_d)'$ is isomorphic to $\clm_0'$.     

\vsp 
More generally, we claim that the commutant of $\Lambda^n(\pi_{\omega_R}(\mbox{UHF}_d))''$ 
is also a type-II$_1$ factor isomorphic to $\clm_0 ' \otimes \{ S_IS_J^*:|I|=|J|=n \}''$. 
That the type-II$_1$ factor $\Lambda^n(\pi_{\omega_R}(\mbox{UHF}_d))'$ contains $\clm_0 ' \otimes \{S_IS_J^*:|I|=|J| \}''$ is obvious. The factor being a hyperfinite type-II$_1$ factor, we may write $\pi_{\omega_R}(\Lambda^n(\mbox{UHF}_d))' = \cln_0 \otimes \{S_IS_J^*:|I|=|J|=n \}''$ for some type-$II_1$ factor $\cln_0$ and $\clm_0' \subseteq \cln_0$. For the reverse inclusion, if $X \in \cln_0$ then $X \in 
\pi_{\omega_R}(\mbox{UHF}_d)'$ and so $X \in \clm'_0$.  

\vsp 
So $\clm_0' \otimes \pi_{\omega}(\mbox{UHF}_d)''$ admits a tracial state and it is a 
type-II$_1$ factor. However, Cuntz relation (16) gives 
$$\bigcap_{n \ge 1} \Lambda^n(\pi_{\omega_R}(\mbox{UHF}_d))'' \subseteq \pi_{\omega_R}(\mbox{UHF}_d)'' \bigcap \pi_{\omega}(\mbox{UHF}_d)'$$ 
So by the factor property of $\omega_R$, we also have 
$\clm'_0 \otimes \pi_{\omega}(\mbox{UHF}_d)''$ is the algebra of all bounded 
operators on $\clh_{\omega_R}$. This brings a contradiction. 

\vsp 
That, $d$ can not be an even integer, is given in \cite{[20]} since 
$\omega_R(\mbox{UHF}_d)''$ is a type-I factor. It also follows by a more general result \cite{[24]}, where we could drop additional assumption that $\omega$ is reflection positive with the twist $r_0$ but here (a) is valid for reflection positive with twist case.   

\vsp 
Thus $d$ is an odd integer and $\omega_R$ is a type-I factor state of $\mbox{UHF}_d$ with its corner $P\pi_{\omega}(\mbox{UHF}_d)''P$ equal to a finite type-I factor $\clm$ and 
by Proposition 2.2 (e) in \cite{[24]} we have 
$$\pi_{\psi}(\clo_d)''=\pi_{\psi}(\mbox{UHF}_d)''$$ 
Since $\beta_z(S_i)=zS_i$ for any $z \in H$ but $\beta_z(X)=X$ for all $X \in \pi_{\psi}(\mbox{UHF}_d)''=\pi_{\psi}(\clo_d)''$. So we have $S_i=zS_i$ 
each $1 \le i \le d$ and $z \in H$. This shows $z=1$ as $S_i^*S_i=I$ for each $1 \le i \le d$. Thus $H=\{1\}$.    

\vsp 
This also shows that $\clm_0=\clm$ since $\clm_0=P\pi_{\psi}(\mbox{UHF}_d)''P$ and $\clm=P\pi_{\psi}(\clo_d)''P$. Thus $Ad_{\gamma_{r_{\zeta}}}=I$ on $\clm$ as well and so 
$$v^*_i=Ad_{\gamma_{r_{\zeta}}}(v^*_i) $$
$$= \clj \beta_{\bar{r_{\zeta}}}(\tilde{v}^*_i) \clj$$
$$= \beta_{r_{\zeta}}(v_i)$$
since $\Delta=I$ and so $\tilde{v}^*_i = \clj v_i \clj$ for $1 \le i \le d$

\vsp 
We are left to prove the last statement (e). The factor $\clm$ being type-I and $SU_2(\IC)$ being simply connected, first part follows by a standard result in representation theory \cite{[18]}. 

\vsp 
We will prove now that the group action $\alpha_g:x \raro \hat{u}(g)x\hat{u}(g)^*$ on $\clm$ is ergodic i.e. there exists no no-trivial invariant element for the group action. Let $(p_i:1 \le i \le m)$ be a maximal set of orthogonal minimal projections in $\clm_G=\{x \in \clm:\alpha_g(x)=x,\;\forall g \in SU_2(\IC) \}$ i.e. elements in $\clm$ that are invariant for the group action $(\alpha_g:g \in SU_2(\IC))$ and $u$ be a unitary element in 
$\clm$ invariant for the group action $\alpha_g$ as well. 

\vsp 
So $V$ as well as $V_{u}=(uv_iu^*)$ satisfies the inter-twinning relation (28). By the uniqueness of Clebs-Gordon coefficients, we get 
\be 
p_iuv_ku^*p_j=c^i_j(u)p_iv_kp_j
\ee
for all $1 \le k \le d$ with some scalers $c^i_j(u) \in \IC$.   

\vsp 
We compute now the following 
\be 
\sum_{1 \le k \le d}p_iuv_ku^*p_juv_k^*u^*p_i = |c^i_j(u)|^2 \sum_{1 \le k \le d}p_iv_kp_jv_k^*p_i
\ee

\vsp 
Since $\sum_k v_kp_jv^*_k$ is also $\alpha_g$-invariant, the left hand side of (56) is independent of $u$, is in the centre of $\clm_G$ and so $|c^i_j(u)|=1$.
By (d), we also have $v_k^*= \beta_{r_0}(v_k)$ and the family of vectors $\{p_iv_kp_j\zeta_{\psi}:1 \le k \le d \}$ are mutually orthogonal for each fix $1 \le i,j \le d$ 
( as $(\phi(p_iv_kp_jv^*_lp_i))$ is a invariant vector for the representation $g \raro u(g)\otimes \overline{u(g)}$ of $SU_2(\IC)$ ), $c^i_j(u)$ is a real number and so either equal to 
$1$ or $-1$. The set of invariant unitary elements in the centre of $\clm_G$ is a connected set and the map $u \raro c^i_j(u)$ is continuous. Thus $c^i_j(u)=1$. So we get $uv_ku^*=v_k$ for all $k$ i.e. $u \in \clm'$. Since $\clm$ is a factor, we conclude that $u$ is a scaler multiple of identity element of $\clm$. This shows that $\clm_G$ is a subfactor of finite type-I factor $\clm$. 

\vsp 
Without loss of generality we write $\clm = \clm'_G \otimes \clm_G$ and $\hat{u}(g)=\hat{u}(g) \otimes I_{\clm_G}$ for all $g \in SU_2(\IC)$ and $\clm'_G = \{\hat{u}(g): g \in SU_2(\IC) \}''$. 

\vsp 
If $u \in \clm_G$ in (56), is only an element that commutes with each $(p_i)$, then the left hand side of (56) is also independent of $u$ since $G$-invariant element $p_i\tau(p_j)p_i$ is a scaler multiple of $p_i$ as each $p_i$ is a minimal projection in $\clm_G$. Now we follow the same argument used above to conclude that $u$ is a scaler multiple of identity operator. Thus $\clm_G$ is trivial. This completes the proof for irreducibility of the representation $g \raro \hat{u}(g)$ of $SU_2(\IC)$.

\vsp 
That the dimension of the representation $g \raro \hat{u}(g)$ could be an odd integer or even integer follows once we appeal to Clebsch-Gordon theorem for 
$$u_s(g) \otimes u_t(g) \equiv u_{|t-s|}(g) \oplus u_{|t-s|+1}(g) \oplus ..\oplus u_{s+t}(g)$$
to verify that for any integer value of $s$, $u_t(g)$ is present in the decomposition of $u_s(g) \otimes u_t(g)$ irrespective of the value $t$ that could be either an integer spin 
or ${1 \over 2}$-odd integer spin representation.  
\end{proof}

\vsp 
Now we sum up our main result of this section in the following theorem with a natural 
generalisation.  

\vsp 
\begin{thm} 
Let $G$ be a simply connected group and $g \raro u^i_j(g)$ is a $d$-dimensional irreducible representation of $G$ such that $g \raro u(g) \otimes \bar{u(g)}$ admits a unique one dimensional invariant subspace in $\IC^d \otimes \IC^d$ and $\omega$ be a real, lattice reflection symmetric with a twist $r_0$, translation invariant pure state of $\IM$. If $\omega$ is also $G$-invariant and reflection positive with the twist $r_0$ then there exists an extremal element $\psi \in K_{\omega}$ so that its associated elements in Proposition 2.5 $(\clk,\clm,v_k:1 \le k \le d )$ satisfies the following:

\NI (a) $\Delta=I$ and $\clm=\clm_0$ is a finite type-I factor; 

\NI (b) $H$ is the trivial subgroup of $S^1$ and $Ad_{\gamma_{r_{\zeta}}}=I$; 

\NI (c) For each $1 \le i \le d$, we have $v_i^* = \beta_{r_{\zeta}}(v_i)$;

\NI (d) Two-point spatial correlation functions of $\omega$ decay exponentially.   

\end{thm}

\vsp 
\begin{proof}
First part of the statement is a simple generalisation of Theorem 4.2 and its proof follows by simple inspection of the proof where we have used those properties of the representation rather than explicit use of it. 

\vsp 
Now we consider the contractive operator $Ta\zeta_{\psi}= \tau(a)\zeta_{\psi}, a \in \clm$, where $\tau(a)= \sum _i v_i a v_i^*,x \in \clm$ and the tracial state $a \raro  \langle \zeta_{\psi}, a \zeta_{\psi} \rangle$ on the finite matrix algebra $\clm$ is invariant for $\tau$. The equality in (c) in particular says that $T$ is also self adjoint and so $T^2$ is positive. 

\vsp 
Thus the exponentially decaying property of two point spatial correlation would be 
related with the mass gap in the spectrum of $T^2$ from $1$ once we show that any inviant vector of $T^2$ is a scaler multiple of $\zeta_{\psi}$. 

\vsp 
Let $f$ be an invariant vector for $T^2$. Then we get 
$$\langle f,  a \zeta_{\psi} \rangle$$
$$=\langle T^{2n} f, a \zeta_{\psi} \rangle$$
$$=\langle f, \tau^{2n}(a) \zeta_{\psi} \rangle$$  
for all $n \ge 1$ and $a \in \clm$. Taking $n \raro \infty$, we conclude
that $$\langle f, a \zeta_{\psi} \rangle = \langle 
\zeta_{\psi}, a \zeta_{\psi} \rangle \langle f, \zeta_{\psi} \rangle$$ 
for all $a \in \clm$ i.e. $f=0$ if $f$ is orthogonal to $\zeta_{\psi}$.

\vsp 
Let $0 \le \delta < 1$ and $\delta^2$ be the highest eigen value of $T^2-|\zeta_{\psi}\rangle \langle \zeta_{\psi}|$ and $\beta > 0$ so that $e^\beta \delta < 1$. So we have 
$||T-|\zeta_{\psi} \rangle \langle \zeta_{\psi}||| \le \delta I$ and for any $A,B \in \IM$
$$e^{\beta n}|\omega(A\theta^n(B))-\omega(A)\omega(B)|$$
$$= e^{\beta n}|\langle a^*\zeta_{\psi}, [T-|\zeta_{\psi}\rangle \langle \zeta_{\psi}|]^n b\zeta_{\psi}\rangle|$$
$$\le (e^\beta \delta)^n ||a|| ||b|| \raro 0$$
as $n \raro \infty$, where $a=P\pi_{\psi}(A)P$ and $b=P\pi_{\psi}(B)P$ are elements in $\clm$ and $e^\beta \delta < 1$.  
 
\end{proof}

\section{Ground states of Hamiltonian in quantum spin chain }

\vsp 
We are left to discuss few motivating examples for this abstract framework, developed so far to study symmetries of Hamiltonian $H$ that satisfies (3) and (14). Before we take few specific examples, we recall some well known results in the following proposition for our reference and its conquences in light of results proved in section 3 and 4.     

\vsp 
\begin{pro} 
Let $H$ be a Hamiltonian in quantum spin chain $\IM= \otimes_{\IZ} \!M_d(\IC)$ 
that satisfies relation (3) with $h_0 \in \IM_{loc}$. Then the following statements are true:

\vsp 
\NI (a) There exists a unique KMS state $\omega_{\beta}$ for $(\alpha_t)$ at each inverse positive temperature $\beta={1 \over kT} > 0$ and $\omega_{\beta}$ is a translation invariant factor state of $\IM$. 

\vsp 
\NI (b) If $H$ also satisfies relation (14) with $J > 0$ and $r_0 \in U_d(\IC)$, then the unique KMS state $\omega_{\beta}$ is reflection positive with twist $r_0$. Furthermore, any weak$^*$ limit point of $\omega_{\beta}$ as $\beta \raro \infty$ is also reflection positive with twist $r_0$; 

\vsp 
\NI (c) If $H$ is also $SU_2(\IC)$ -invariant i.e. $\beta_{u(g)}(h_0)=h_0$ and $\omega$ be a low temperature limit point 
ground state for $H$ described in (b) then $\omega \in S_{\theta,SU_2(\IC) \otimes \IZ_2,+}$ and $\omega=\int^{\oplus}_{X} \omega^{\alpha}d\mu(\alpha)$ be its extremal decomposition in $S_{\theta,SU_2(\IC) \otimes \IZ_2,+}$. Then the following hold:

\NI (i) For any odd integer $d \ge 3$, $\mu$-allmost everywhere, $\omega^{\alpha}$ are pure 
ground states of $H$. 
 
\NI (ii) If $H$ is also real then for any even integer $d \ge 2$, ground states $\omega^{\alpha}$ are not even  extremal elements in the convex set of translation invariant states of $\IM$ for a $\mu$-positive Borel set
of $\alpha$.  
 
\end{pro}

\begin{proof} 
For (a), we refer to H. Araki's work \cite{[4]} and also \cite{[17]}. For the first statement in (b), we refer to \cite{[12]}. Last part of (b) is trivial as reflection positive property (11) is closed under weak$^*$ limit. 

\vsp 
For the statement of (c), we use Therem 3.7. In particular, $SU_2(\IC)$-invariance ensures that $\mu$-almost everywhere $\omega^{\alpha}$ are stationary states for the Hamitonian dynamics $H$ for which $\beta_{u(g)}(h_0)=h_0$ for all $g \in SU_2(\IC)$ and so $\mu$-almost everywhere $\omega^{\alpha}$ are as well ground states of $H$, since the set of ground states is a face in the convex set of stationary states \cite{[7]}. 

\vsp 
We use now a standard fact that factor decomposition coincides with extremal decompostion for a ground state of a Hamiltonian $H$ \cite{[7]} and so extreme points are pure states of $\IM$ for odd values of $d$. For even values of $d$, these extreme points are not extremal elements in the convex set of translation invariant states of $\IM$ as otherwise these states would have been factor states of $\IM$ by Proposition 3.6 (a) and so we would have been pure states of $\IM$, contradicting the fact that there exists no pure state that is real, lattice symmetric and
$SU_2(\IC)$-invariant \cite{[24]}.   
\end{proof} 

\vsp
\begin{proof} (Theorem 1.3 ) Proof for (a) is given Proposition 5.1 (c) and for (b) we recall that there is no real, lattice symmetric, $SU_2(\IC)$ and translation invariant pure state of $\IM$ for even values of $d$ \cite{[24]}.
\end{proof}

\vsp 
For an even integer, such extremal elements $\omega_{\alpha} \in S_{\theta,G,+}$ in the decompostion given in (ii) of Proposition 5.1 (c) is also real but far from being extremal in the convex set of translation invariant states of $\IM$.  Nevertheless, we have
$$4\omega_{\alpha} = \omega'_{\alpha} + \bar{\omega'_{\alpha}},$$
where $\omega'_{\alpha} = \omega^1_{\alpha} + \tilde{\omega^1}_{\alpha} \beta_{r_{\zeta}}$
for some translation invariant ergodic states $\omega^1_{\alpha}$ of $\IM$.
In particular, this shows that translation invariant ergodic state $\omega^1_{\alpha}$ are $SU_2(\IC)$ invariant ground states of $H$ for $\mu$-almost everywhere but fails to be reflection positive with twist $\beta_{r_{\zeta}}$. In other words, spontaneous $\IZ_2 \times \IZ_2$ symmetries
$\omega \raro \bar{\omega}$ or $\omega \raro \tilde{\omega}$ breaks down rather than $SU_2(\IC)$ symmetry \cite{[24]} if these extremal states in $S_{\theta,\IZ_2,+}$ given in Proposition 5.1 (c) for even values of $d \ge 2$
are decomposed further into translation invaiant ergodic states. This feature is
a stricking contast to the classic case of Ghosh Mazumdar model [GM] that fails to be reflection positive with the twist $\beta_{r_0}$. We end this section with the well known example and compare with our main results of this paper.

\vsp
\begin{exam} Ghosh-Majumdar Model \cite{[13]}: The following well known model with $J > $
$$h_0^{GM}=J(\sigma_x^{(0)} \otimes \sigma_x^{(1)} +\sigma_y^{(0)} \otimes \sigma_y^{(1)} +\sigma_z^{(0)} \otimes \sigma_z^{(1)} + {J \over 2}(\sigma_x^{(0)} \otimes \sigma_x^{(2)} + \sigma_y^{(0)} \otimes \sigma_y^{(2)}+\sigma_z^{(0)} \otimes \sigma_z^{(2)}$$
admits two fold degeneracy in its ground states i.e. the model has two pure ground states for $d=2$. These two pure states are $SU_2(\IC)$ invariant but not translation invariant. However their mean state is translation invariant and extremal in the convex set of all translation invariant state. The mean state being the unique state that admits translation and $SU_2(\IC)$ symmetry, it is the unique low temperature limiting state. The Hamiltonian $H^{GM}=\sum \theta^n(h_0^{GM})$ being not of the form given in (14).  The unique transition invaiant ground state is an ergodic but fails to be a factor state. The state is not reflection positive with twist $\beta_{r_0}$. Thus Propsotion 5.1 is
not valid without our assumption that $\omega$ is reflection positive with twist $r_0$.
\end{exam}

\section{Haldane's conjecture:} 

\vsp 
In the last section we consider Heisenberg anti-ferromagnetic model $H^{XXX}$
model with odd integer $d=2s+1$ i.e. integer degrees of freedom $s$ for spin chain
electrons placed in a one dimensional lattice $\IZ$. We will also discuss briefly
Heisenberg anti-feromagnetic model $H_{XXX}$ on higher lattice dimension.

\vsp
If $H_{XXX}$ admits unique ground state then the ground state $\omega_{XXX}$ is pure, translation invariant, $SU_2(\IC)$-invariant, reflection symmtric with twist and positive. Theorem 4.5 says that such a ground state is also finitely corelated and its spatial corelation functions decay exponentially. The following statement is an easy consequence of standard results \cite{[8],[19]}.

\vsp 
\begin{thm} 
Let $\omega_{\beta}$ be the unique thermal equilibrium or KMS factor state at inverse temperature $\beta$ for anti-ferromagnetic $H^{XXX}$ model with odd integer $d=2s+1 \ge 3$ 
($s$ is an integer greater than equal to $1$) and $\omega$ be a limit point of $\omega_{\beta}$ as $\beta \raro \infty$. Then following holds:

\vsp 
\NI (a) Then $\omega=\int\omega_r d\mu(r)$, where $\omega_r$ is the state defined by 
$$\omega_r(e^{i_1}_{j_1} \otimes ...\otimes e^{i_n}_{j_n})=\phi(v_Iv_J^*)$$ 
and $v=(v_i)$ is the unique solution to Clebsch-Gordon equation (52) as described in Theorem 4.2 satisfying (d) and (e) with irreducible representations $g \raro u(g)=u_s(g)$ and $g \raro \hat{u}(g)=u_r(g)$ of $SU_2(\IC)$ with finite $I_{2r+1}$ factor $\clm$ for an integer $r \ge 1$ or half-odd integer. 

\vsp 
\NI (b) In (a) the dimension of $\clm$ i.e. irreducible representation $g \raro \hat{u}_r(g)$ in Theorem 4.2 with dimension $2r+1$ with half odd-integer $r$ or integer $r$ is determined by minimising mean energy of $H^{XXX}$ over all possible solutions to (52) with irreducible representations $g \raro \hat{u}_{l}(g)$ of dimension $2l+1$ with half odd-integer or integer ( each $\omega_l$ is an invariant state for Hamiltonian flow $\hat{\sigma}_t$ of $H^{XXX}$ ) i.e.
\be 
\omega_r(h_0)=\mbox{min}_{l={1\over 2},1, {3 \over 2},.. } \omega_l(h_0)
\ee

\vsp 
\NI (c) If there exist unique $r$ for which $\omega_l(h_0^{xxx})$ attains its minimum then 
$\omega=\omega_r$ i.e. the low temperature limit of $\omega_{\beta}$ as $\beta \raro \infty$ is unique and its limiting value is $\omega_r$.

\end{thm} 

\vsp 
We illustrate our results for $d=2s+1=3$ in the following text for possible further investigation.  

\vsp 
Now we briefly discuss the situation when $d=3$ i.e. $s=1$. In such a case Pauli spin matrices are given by  
\ben
\sigma_x= 2^{-{1 \over 2}} \left (\begin{array}{llll} 0&,&\;\; 1,\;\; 0 \\ 1&,&\;\;0,\;\;1 \\ 0&,&\;\;1,\;\;0
\end{array} \right ),
\een
\ben
\sigma_y = 2^{-{1 \over 2}} \left (\begin{array}{llll} 0&,&\;\; -i,\;\; 0 \\ i&,&\;\;0,\;\;-i \\ 0&,&\;\;i,\;\;0
\end{array} \right ),
\een
\ben
\sigma_z= \left (\begin{array}{llll} 1&,&\;\; 0,\;\; 0 \\ 0&,&\;\;0,\;\;0 \\ 0&,&\;\;0,\;\;-1
\end{array} \right ).
\een
and $i\sigma_x,i\sigma_y,i\sigma_z$ are basis for Lie-algebra $su_2(\IC)$ with 
$$[i\sigma_x,i\sigma_y]=-i\sigma_z,\;\;[i\sigma_y,i\sigma_z]=-i\sigma_x,\;\;[i\sigma_z,i\sigma_x]=-i\sigma_y$$
\vsp 
A direct calculation shows that the inter-twiner $r_{\zeta}=r_0$ is a matrix with real entries given below 
\ben
r_{\zeta} = \left (\begin{array}{llll} 0&,&\;\; 0,\;\; -1 \\ 0&,&\;\;1,\;\;0 \\ -1&,&\;\;0,\;\;0
\end{array} \right ).
\een 
and 
$$h^{xxx}_0= J (\sigma_x \otimes \sigma_x + \sigma_y \otimes \sigma_y + \sigma_z \otimes \sigma_z)$$

\vsp 
By Clebsch-Gordon decomposition of $SU_2(\IC)$ representation $g \raro u_1(g) \otimes u_1(g)$, the 
commutant of $\{u_1(g) \otimes u_1(g):g \in SU_2(\IC) \}$ in $\IM_3(\IC) \otimes \IM_3(\IC)$ is equal to its centre made of orthogonal projections of dimension $1,3,5$. Since $h_0$ commutes with $u_1(g) \otimes u_1(g)$, $h_0$ is in the centre of $\{u_1(g) \otimes u_1(g): g \in SU_2(\IC)\}''$ and so $\omega_l(x h_0) = \omega_l(h_0x)$ for any $x \in \IM_3(\IC) \otimes \IM_3(\IC)$ since $\omega_l$ is $SU_2(\IC)$-invariant of $\IM$. Thus $\omega_l \sigma^{XXX}_t = \omega_l$ for all $t \in \IR$ 
on local elements of $\IM$ and so on $\IM$.

\vsp 
We compute with $J=1$
$$\omega_{l}(h^{xxx}_0) = {1 \over 2} \phi(v_1av^*_2+
v_2av^*_1 + v_3av^*_2 +v_2av^*_3))$$
$$\;\;\;\;\;\;\;\;\;\;\;\;\;-{1 \over 2}\phi(-v_1bv_2^*+v_2bv_1^*+v_3bv_2^*-v_2bv_3^*)$$
$$+\phi(v_1cv^*_1-v_3cv^*_3)$$
(where $a=v_1v^*_2+v_2v_1^*+v_3v^*_2+v_2v_3^*$, $b=-v_1v_2^*+v_2v_1^*+v_3v^*_2-v_2v_3^*$ 
and $c=v_1v^*_1-v_3v^*_3$)
$$={1 \over 2} \phi( (v^*_2v_1+v_1^*v_2+v_3^*v_2+v_2^*v_3)(v_1v_2^*+v_2v_1^*+v_3v_2^*+v_2v_3^*))$$
$$\;\;\;\;\;\;\;\;\;\;\;\;\;-{1 \over 2}\phi( (-v^*_2v_1+v_1^*v_2+v_2^*v_3-v_3^*v_2)(-v_1v_2^*+v_2v_1^*+v_3v_2^*-v_2v_3^*))$$
$$+\phi((v_1^*v_1-v_3^*v_3)(v_1v^*_1-v_3v^*_3))$$
(where we have tracial state property of $\phi$ on $\clm$ ) 
$$=\phi((v_2^*v_1+v_3^*v_2)(v_2v_1^*+v_3v_2^*))$$
$$+\phi((v_1^*v_2+v_2^*v_3)(v_1v_2^*+v_2v_3^*))$$
$$+\phi((v_1^*v_1-v_3^*v_3)(v_1v_1^*-v_3v_3^*))$$

\vsp 
Since $v_1^*=-v_3,\;v_2^*=v_2$ and $v_1v_1^*+v_2v_2^*+v_3v_3^*=I$, we simplify further that
$$\omega_{t}(h^{xxx}_0)=\phi(\delta_{v_2}(v_1)\delta_{v_2}(v_1^*)) + \phi(\delta_{v_2}(v_3)\delta_{v_2}(v_3^*))$$
\be 
+ \phi(((v_1^*v_1-v_3^*v_3)(v_1v_1^*-v_3v_3^*))
\ee
where we have used the symbol $\delta_{v_2}(a)=v_2a-av_2$ for $a \in \clm$. 

\vsp 
Solution to (52) is given by $v_1=l_+,v_3=l_-$ and $v_2=il_z$, where 
$$\sqrt{l(l+1)} l_+ = {1 \over \sqrt{2}}(\bar{\pi}_l(i\sigma_x)+i\bar{\pi}_l(i\sigma_y)),$$
$$\sqrt{l(l+1)} l_- = {1 \over \sqrt{2}}(\bar{\pi}_l(i\sigma_x)-i\bar{\pi}_l(i\sigma_y))$$
and
$$\sqrt{l(l+1)} l_z= \bar{\pi}_l(i\sigma_z)$$
where we used notation $\bar{\pi}_l$ for $2l+1$-dimensional irreducible representation of Lie-algebra $su_2(\IC)$ of the Lie-group $SU_2(\IC)$.  So 
$$l(l+1)\delta_{v_2}(v_1)$$
$$={i \over \sqrt{2}}( [\bar{\pi}_l(i\sigma_z),\pi_l(i\sigma_x)] +i [\bar{\pi}_l(i\sigma_z),\bar{\pi}_l(i\sigma_y)])$$
$$={i \over \sqrt{2}}(\bar{\pi}_l(-i\sigma_y)+i(\bar{\pi}_l(i\sigma_x))$$
$$={1 \over \sqrt{2}}( -i\bar{\pi}_l(i\sigma_y) - \bar{\pi}_l(i\sigma_x)$$
$$=-\sqrt{l(l+1)}v_1$$
i.e 
$$\sqrt{l(l+1)}\delta_{v_2}(v_1)=-v_1$$

\vsp 
Similarly, we may compute by taking ajoint that 
$$\sqrt{l(l+1)}\delta_{v_2}(v_3)=v_3$$
We also compute that 
$$l(l+1)(v_1v_1^*-v_3v_3^*)$$
$$=-l(l+1)(v_1v_3-v_3v_1)$$
$$=-{ 1 \over 2} [\pi_l(i\sigma_x)+i\pi_l(i\sigma_y),\pi_l(i\sigma_x)-i\pi_l(i\sigma_y)]$$
$$=i[\pi_l(i\sigma_x),\pi_l(i\sigma_y)]$$
$$=i\bar{\pi}_l(-i\sigma_z)$$
$$=-i\bar{\pi}_l(i\sigma_z)$$
i.e. $\sqrt{l(l+1)}(v_1v_1^*-v_3v_3^*) = -v_2$

\vsp 
Now from (59), we get 
$$\omega_l(h^{xxx}_0)= -{1 \over l(l+1)} \phi(v_1v_1^*+v_3v_3^*) - { 1 \over l(l+1)} \phi(v_2v^*_2)$$
$$= -{ 1 \over l(l+1)}\phi(v_1v_1^*+v_2v_2^*+v_3v_3^*)$$
$$=-{ 1 \over l(l+1)}$$

\vsp 
The above computation could have been simplified by using $SU_2(\IC)$ symmtry of the state 
$\omega_l$ to write 
$$\omega_l(h^{xxx}_0)=3\omega_l(\sigma_z \otimes \sigma_z)$$
$$=-{3 \over l(l+1)} \omega(v_2v_2^*)$$
$$=-{ 1 \over l(l+1)}$$

\vsp 
The map $l \raro \omega_l(h^{xxx}_0)$ increases strictly to zero as $l={1 \over 2},1, {3 \over 2},...$ increases to infinity and its minimum value at $l={1 \over 2}$ is $-{4 \over 3}$. This shows that there is an unique limit point in $\omega_{\beta}$ as $\beta \raro \infty$ and limiting value is $\omega_{1 \over 2}$. 

\vsp 
\begin{thm}
Anti-feromagnetic $H^{XXX}$ model with $d=3$ admits unique low temperature limiting ground state and the state is given by $\omega_{1 \over 2}$. 
\end{thm} 

\vsp 
For any arbitary odd values of $d$, for uniqueness of low temperature limit points, we need to prove that the mean energey of $H^{xxx}$ i.e. $\omega_l(h^{xxx}_0)$ gets minimised by a unique state $\omega_r$ for some ${1 \over 2}$-odd integer or integer spin $r$. Uniqueness of low temperature limiting states as well holds for any odd integer $d=2s+1$ if 
the function $l \raro \omega_l(h^{xxx}_0)$ is a monotonically increasing function in the varaiable $l$. For a possible quick proof, we verify using $SU_2(\IC)$ symmetry that 
$${1 \over 3} \omega_l(h^{xxx}_0)$$
$$=\omega_l(\sigma_z \otimes \sigma_z)$$  
$$=\phi_l(\sum_{1 \le k,k' \le s} k k' ((v^*_{k'}v_{k'}-v_{2s+1-{k'}}^*v_{2s+1-{k'}})(v_kv_k^*-v_{2s+1-k}v_{2s+1-k}^*)),$$
 where $\phi_l$ is the normalized trace on $\IM_{2l+1}$ and $V^*=(v^*_k)$ is the Clebsch-Gordon isometry that intertwins two representations $\pi_l \otimes \pi_s$ and $\pi_l$ of 
$SU_2(\IC)$. Note that there exists a unique intertwiner isometry provided $l \ge |l-s|$ i.e. $l \ge {s \over 2}$. So we leave it for future investigation as conjecture that low temperature limiting ground state of $H_{XXX}$ for odd values of $d=2s+1$ is 
$\omega_{s \over 2}$. 

\vsp 
For further illustration of our main result, we consider now well-studied AKLT model $H^{AKLT}$ [1] for which 
$$h^{aklt}_0= J({1 \over 3} + {1 \over 2}(h^{xxx}_0 + {1 \over 3}(h^{xxx}_0)^2)$$
It is well known that $\omega_{1 \over 2}$ is the unique ground state for $H^{AKLT}$ with $\omega_{1 \over 2}(h^{aklt}_0)=0$. 

\vsp  
We may as well compute 
$$\omega_l(\sigma_x^2 \otimes \sigma_x^2+\sigma_y^2 \otimes \sigma_y^2+\sigma_z^2 \otimes \sigma_z^2)$$
$$={1 \over 4}\phi((v_1v_1^*+v_1v_3^*+2v_2v_2^*+v_3v_1^*+v_3v_3^*)(v_1^*v_1+v_3^*v_1+2v_2^*v_2+v_1^*v_3+v_3^*v_3))$$
$$+{1 \over 4}\phi((v_1v_1^*-v_1v_3^*+2v_2v_2^*-v_3v_1^*+v_3v_3^*)(v_1^*v_1-v_3^*v_1+2v_2^*v_2-v_1^*v_3+v_3^*v_3))$$
$$+\phi((v_1v_1^*+v_3v_3^*)(v_1^*v_1+v^*_3v_3))$$
$$={1 \over 4}\phi((I+v_1v_3^*+v_2v_2^*+v_3v_1^*)(I + v_3^*v_1+v_2^*v_2+v_1^*v_3))$$
$$+{1 \over 4}\phi((I -v_1v_3^*+v_2v_2^*-v_3v_1^*)(I -v_3^*v_1+v_2^*v_2-v_1^*v_3))$$
$$+\phi((I-v_2v_2^*)(I-v_2^*v_2))$$
$$={1 \over 2}\phi((I+v_2v_2^*)(I +v_2^*v_2))
  +{1 \over 2}\phi((v_1v_3^*+v_3v_1^*)(v_3^*v_1+v_1^*v_3))$$
$$+\phi((I-v_2v_2^*)(I-v_2^*v_2))$$
$$={3 \over 2}(1+\phi(v_2^4)) -\phi(v_2^2) +{1 \over 2}\phi((v_1v_3^*+v_3v_1^*)(v_3^*v_1+v_1^*v_3))$$

\vsp 
We also compute $$l(l+1)(v_1v_3^*+v_3v_1^*)$$
$$={1 \over 2}\{ (\bar{\pi}_l(i\sigma_x)+i\bar{\pi}_l(i\sigma_y)) (-\bar{\pi}_l(i\sigma_x)- i\bar{\pi}_l(i\sigma_y))+(\bar{\pi}_l(i\sigma_x)-i\bar{\pi}_l(i\sigma_y))(-\bar{\pi}_l(i\sigma_x)+ i\bar{\pi}_l(i\sigma_y))\}$$
$$=\bar{\pi}_l(i\sigma_y)^2-\bar{\pi}_l(i\sigma_x)^2$$
So 
$$l^2(l+1)^2\phi((v_1v_3^*+v_3v_1^*)(v_3^*v_1+v_1^*v_3))$$
$$=\phi((\bar{\pi}_l(i\sigma_y)^2 - \bar{\pi}_l(i\sigma_x)^2)^2)$$
$$=\phi(\bar{\pi}_l(i\sigma_y)^4 + \bar{\pi}_l(i\sigma_x)^4) - 2\phi(\bar{\pi}_l(i\sigma_x)^2\bar{\pi}_l(i\sigma_y)^2)$$
($\phi$ being tracial state on $\clm$, $\phi(\bar{\pi}_l(i\sigma_y)^2\bar{\pi}_l(i\sigma_x)^2)=\phi(\bar{\pi}_l(i\sigma_x)^2 \bar{\pi}_l(i\sigma_y)^2$)

\vsp 
Using symmetry and tracial property of $\phi$ on $\clm$, we write 
$$\alpha_l=\phi(\bar{\pi}_l(i\sigma_x)^4)=\phi(\bar{\pi}_l(i\sigma_y)^4)=\phi(\bar{\pi}_l(i\sigma_z)^4)$$
and 
$$\beta_l=\phi(\bar{\pi}_l(i\sigma_x)^2 \bar{\pi}_l(i\sigma_y)^2) =\phi(\bar{\pi}_l(i\sigma_y)^2 \bar{\pi}_l(i\sigma_z)^2)= \phi(\bar{\pi}_l(i\sigma_z)^2 \bar{\pi}_l(i\sigma_x)^2)$$
We use the following identities: 
$$l^2(l+1)^2 = \phi( (\bar{\pi}_l(i\sigma_x)^2 + \bar{\pi}_l(i\sigma_y)^2 + \bar{\pi}_l(i\sigma_z)^2)^2)$$
$$=3\alpha_l + 6 \beta_l$$
and 
$$\alpha_l = {1 \over 2l+1}\sum_{-l \le m \le l} m^4$$
to deduce 
$$\omega_l(\sigma_x^2 \otimes \sigma_x^2+\sigma_y^2 \otimes \sigma_y^2+\sigma_z^2 \otimes \sigma_z^2)$$
$$={3 \over 2} -{1 \over 3} + {1 \over l^2(l+1)^2} ({3\alpha_l \over 2} + \alpha_l-\beta_l)$$
$$={7 \over 6} + {1 \over l^2(l+1)^2}( {5\alpha_l \over 2} - {1 \over 6}(l^2(l+1)^2-3\alpha_l))$$
$$=1 + {3\alpha_l \over l^2(l+1)^2} $$

\vsp 
We can as well use symmetry and tracial state property of $\phi$ to compute 
$$\omega_l((h_0^{xxx})^2) = 3 \omega_l(\sigma^2_z \otimes \sigma^2_z) + 6 \omega_l(\sigma_z \sigma_x \otimes \sigma_z\sigma_x)$$
$$=3 \phi((I-v_2^2)(1-v_2^2)) + 3 \phi((v_3v_1^*-v_1v_3^*)(v_1^*v_3-v_3^*v_1))$$
$$= 1 + 3 \phi(v_2^4) + 6 \phi(v_2^4)$$
$$= 1 + 9 \phi(v_2^4)$$

So 
$$\omega_l(h_0^{aklt}) = {1 \over 3} + { 1 \over 2}({1 \over 3}+ 3 \phi(v_2^4)-{1 \over l(l+1)})$$ 
$$= {1 \over 3} + {1 \over 2}({1 \over 3} - {1 \over l(l+1)}  + {3 \over l^2(l+1)^2(2l+1)} \sum_{-l \le m \le l}m^4)$$

\vsp 
This clearly shows that the unique ground state $\omega=\omega_{1 \over 2}$ of $H^{AKLT}$ \cite{[1]} is also the low temperature limiting ground state of $H^{XXX}$. So we have the following well known result \cite{[1]}. 

\vsp 
\begin{thm} 
Low temperature limit point as $\beta \raro \infty$ of the unique temperature states $\omega_{\beta}$ for $H^{AKLT}$ at inverse temperatures $\beta > 0$ is unique and its 
limiting value is also given by $\omega_{1 \over 2}$.
\end{thm}

\vsp
\begin{thm}
Let $\omega$ be a low temperature limiting ground state of anti-feromagnetic Heisenberg nearest neighbour isospin model
$$H_{XXX}=\sum_{|\ul{i}-\ul{j}|=1} \sigma^{\ul{i}}_{x} \otimes \sigma^{\ul{j}}_x + \sigma^{\ul{i}}_y \otimes \sigma^{\ul{j}}_y+\sigma^{\ul{i}}_z \otimes \sigma^{\ul{j}}_z$$
on higher lattice dimension $\IZ^q= \IZ \otimes \IZ \otimes \IZ$ for $q \ge 2$. Then $\omega \in S_{\theta,G,+}$, where $S_{\theta,G,+}$ are defined as natural generalisation in higher lattice dimension with reflection symmetries of lower dimensional lattices.  We consider extremal decomposition of $\omega=\int \omega_{\alpha} d\mu(\alpha)$ in $S_{\theta,G,+}$. Then the folloing statements are true:

\vsp
\NI (a) If $d$ is an odd integer then extremal elememts $\omega_{\alpha}$ in the decomposition are pure.

\vsp
\NI (b) If $d$ is an even integer then extremal elememts $\omega_{\alpha}$ in the decomposition are not even extremal in convex set of translation invariant states.

\end{thm}

\vsp
\begin{proof}
Going alonng the same line used in the proof for Proposition 3.6 (b) we deduce that $\omega_{\alpha}$ are extremal elements in the convex set of transition invariant states using reflection positivity around lower dimensional lattice. That these ergodic states are factors for odd values of $d$ needs additional argument. We will use induction on lattice dimension. We already proved the statement for one lattice dimension. We recall Power's criteria \cite{[29]} and use standard approximation to note that factor proprty of $\omega$ is equivalent to show factor propery of $\omega_{Y}$ for all finite subset $Y$ of $\IZ^2$, where $\omega_Y$ is the restriction of $\omega$ to $\IM_Y$. Fix any finite subset $Y$ of $\IZ^2$, we find
an integer $m \ge 1$ such that $Y \subset \IZ \times \{k: -m \le k \le m\}$. Since the state $\omega$ restricted to $\IM_{\IZ \times \{k: -m \le k \le m \}}$ is an ergodic state and reflection positive with twist, is a factor state. Thus $\omega$ restricted to $\IM_{Y}$ is also a factor state.

\vsp

\end{proof}

\vsp

\section{Disclaimer} 
 
Present manuscript being on pure mathematics, to my knowledge, I have no conflict
of interest with scientific community.
\section{Data Availability} 

Data sharing is not applicable to this article as no new data were created or analyzed
in this study.

\end{document}